\documentclass[11pt,letterpaper]{article}
\usepackage{times,mathptmx}
\DeclareMathAlphabet{\mathcal}{OMS}{cmsy}{m}{n}
\usepackage{fullpage}
\usepackage{amsmath,amsfonts,amssymb,amsthm,amscd}
\usepackage{hyperref,cite,url}
\hypersetup{pdfpagemode=UseNone,pdfstartview=}
\newcommand{\vect}[1]{\mathbf{#1}}

\newcommand{\vs}{\vspace{1.5mm}}
\theoremstyle{plain} 
\newtheorem{theorem}{Theorem}[section]
\newtheorem{lemma}[theorem]{Lemma}

\theoremstyle{definition} 
\newtheorem{definition}{Definition}[section]
\newtheorem{assumption}{Assumption}
\theoremstyle{remark} 
\newtheorem{remark}{Remark}
\newcommand{\G}{\mathbb{G}}
\newcommand{\Z}{\mathbb{Z}}
\newcommand{\bits}{\{0,1\}}
\newcommand{\Adv}{\textbf{Adv}}
\newcommand{\mc}[1]{\mathcal{#1}}
\newcommand{\tb}[1]{\textbf{#1}}
\newcommand{\lb}{\linebreak[0]}
\newcommand{\db}{\displaybreak[0]}

\title{Transforming Hidden Vector Encryption Schemes\\ from Composite to
Prime Order Groups}

\author{
    Kwangsu Lee\footnote{Sejong University. Seoul, Korea.
        Email: \texttt{kwangsu@sejong.ac.kr}.}
}

\date{}

\begin{document}

\maketitle

\begin{abstract}
Predicate encryption is a new type of public key encryption that enables
searches on encrypted data. By using predicate encryption, we can search
keywords or attributes on encrypted data without decrypting ciphertexts.
Hidden vector encryption (HVE) is a special kind of predicate encryption.
HVE supports the evaluation of conjunctive equality, comparison, and subset
operations between attributes in ciphertexts and attributes in tokens. In
this paper, we construct efficient HVE schemes in prime order bilinear
groups derived from previous HVE schemes in composite order bilinear
groups, and prove their selective security under simple assumptions. To
achieve this result, we present a conversion method that transforms HVE
schemes from composite order bilinear groups into prime order bilinear
groups. Our method supports any types of prime order bilinear groups and
uses simple assumptions.
\end{abstract}

\vs \noindent {\bf Keywords:} Searchable encryption, Predicate encryption,
Hidden vector encryption, Conversion method, Bilinear maps.

\newpage

\section{Introduction}

Searchable public key encryption is a new type of public key encryption (PKE)
that enables efficient searching on encrypted data \cite{BonehCOP04}. In PKE,
if an agent $A$ wants to search on encrypted data for a user $B$, he should
first decrypt ciphertexts using the private key $SK$ of the user $B$. This
simple method has a problem that the agent requires the user's private key.
In searchable public key encryption, a ciphertext is associated with keywords
or attributes, and a user can generate a token for searching from the user's
private key. That is, an agent $A$ performs searches on encrypted data using
the token $TK$ that is related with keywords or attributes instead of using
the private key $SK$. By using searchable public key encryption, it is
possible to build interesting systems like privacy preserving mail gateway
systems \cite{BonehCOP04}, secure audit log systems \cite{WatersBDS04},
network audit log systems \cite{ShiBCSP07}, and credit card payment gateway
systems \cite{BonehW07}.

Predicate encryption (PE) is a generalization of searchable public key
encryption \cite{BonehW07,KatzSW08}. In PE, a ciphertext is associated with
an attribute $x$, and a token is associated with a predicate $f$. At first, a
sender creates a ciphertext that is associated with an attribute $x$, and an
agent receives a token that corresponds to a predicate $f$ from a receiver.
If $f(x)=1$, then the agent can decrypt ciphertexts that are related with
$x$. Otherwise, that is $f(x)=0$, then the agent cannot get any information
except that $f(x)=0$. That is, PE provides both \textit{message hiding} and
\textit{attribute hiding} properties. Hidden vector encryption (HVE) is a
special kind of PE \cite{BonehW07}. In HVE, a ciphertext and a token are
associated with attribute vectors $\vect{x}, \vect{y}$ respectively, and the
attribute vector for the token contains a special wild card attribute. If
each attribute of a ciphertext is equal with the attribute of a token except
the wild card attribute, then the predicate $f_{\vect{y}}(\vect{x})$ is
satisfied. HVE supports the evaluation of predicates such that conjunctive
equality, conjunctive subset, and conjunctive comparison.

Many HVE schemes were originally proposed in composite order bilinear groups
\cite{BonehW07,ShiW08,LeeL11}. To improve the efficiency of HVE schemes, HVE
schemes in prime order bilinear groups are required. Although many HVE
schemes in prime order groups were constructed from scratch \cite{IovinoP08,
OkamotoT09,Park11e}, we would like to easily obtain HVE schemes in prime
order groups from previous schemes in composite order groups. The previous
conversion methods that convert cryptographic schemes from composite order to
prime order bilinear groups are Freeman's method \cite{Freeman10} and Ducas'
method \cite{Ducas10}. The method of Ducas is that random blinding elements
in ciphertexts can be eliminated in asymmetric bilinear groups of prime order
since the decisional Diffie-Hellman (DDH) assumption holds in asymmetric
bilinear groups. The method of Freeman is that product groups and vector
orthogonality provide the subgroup decision assumption and the subgroup
orthogonality property in prime order bilinear groups, respectively. The
merit of this method is that it can convert many cryptographic schemes from
bilinear groups of composite order to asymmetric bilinear groups of prime
order. The demerits of this method are that the converted scheme only works
in asymmetric bilinear groups and the security of the scheme is proven under
complex assumptions.

\subsection{Our Results}

In this paper, we present a new conversion method that transforms HVE schemes
from composite order bilinear groups into prime order bilinear groups.

Our conversion method is similar to the conversion method of Freeman
\cite{Freeman10} since it uses product groups and vector orthogonality, but
ours has the following three differences. The first difference is that
Freeman's method is related to the subgroup decision (SD) assumption in prime
order bilinear groups, whereas our method is not related to the SD
assumption. The second difference is that Freeman's method only works in
asymmetric bilinear groups of prime order, whereas our method works in any
bilinear groups of prime order. The third difference is that cryptographic
schemes that are converted from Freeman's method use complex assumptions that
depend on complex basis vectors, whereas HVE schemes that are converted from
our method use simple assumptions that are independent of basis vectors.

By using our conversion method, we first convert the HVE scheme of Boneh and
Waters \cite{BonehW07} in composite order bilinear groups into an HVE scheme
in symmetric bilinear groups of prime order. We then prove the converted HVE
scheme is selectively secure under the decisional bilinear Diffie-Hellman
(DBDH) and the parallel 3-party Diffie-Hellman (P3DH) assumptions. Next, we
also convert the delegatable HVE scheme of Shi and Waters \cite{ShiW08} and
the efficient HVE scheme of Lee and Lee \cite{LeeL11} from composite order
bilinear groups to HVE schemes in symmetric bilinear groups of prime order.
Finally, we show that the new P3DH assumption holds in generic group model
introduced by Shoup.

\subsection{Related Work}

PE is closely related to functional encryption \cite{BonehSW11}. In
functional encryption, a ciphertext is associated with attributes $\vect{x}$,
and a private key is associated with a function $f$. The main difference
between PE and functional encryption is that the computation of a predicate
$f(x) \in \bits$ is only allowed in PE whereas the computation of any
function $f(x)$ is allowed in functional encryption. Identity-based
encryption (IBE) is the most simple type of functional encryption, and it
provide an equality function for an identity in ciphertexts \cite{BonehF01}.
Hierarchical IBE (HIBE) is an extension of IBE, and it provides a conjunctive
equality function for a hierarchical identity in ciphertexts
\cite{GentryS02}. Attribute-based encryption (ABE) is also an extension of
IBE, and it provides the most general function that consists of AND, OR, NOT,
and threshold gates \cite{GoyalPSW06}.

The first HVE scheme was proposed by Boneh and Waters \cite{BonehW07}. After
their construction, various HVE schemes were proposed in \cite{ShiW08,
Ducas10,LeeL11}.
A simple HVE scheme can be constructed from a PKE scheme \cite{BonehCOP04,
BonehW07,KatzY09}. This method was introduced by Boneh et al.
\cite{BonehCOP04} to construct a PKE scheme with keyword search (PEKS) using
trapdoor permutations. After that, Boneh and Waters showed that a searchable
public key encryption for general predicates also can be constructed from
this method \cite{BonehW07}. Katz and Yerukhimovich \cite{KatzY09} showed
that it is possible to construct a PE scheme from a PKE scheme if the number
of predicate is less than a polynomial number of a security parameter. The
main idea of this method is to use a multiple instances of key-private PKE
introduced by Bellare et al. \cite{BellareBDP01}. That is, the public key of
searchable public key encryption consists of the public keys of key-private
PKE and each instance of public keys is mapped to each predicate. However,
this method has a serious problem that the total number of predicates is
limited to the polynomial value of a security parameter.

Another HVE scheme can be constructed by extremely generalizing anonymous IBE
(AIBE) \cite{BonehW07,ShiW08,IovinoP08,Ducas10,LeeL11,Park11e}. This method
was introduced by Boneh and Waters \cite{BonehW07}. They used the IBE scheme
of Boneh and Boyen \cite{BonehB04e} and composite order bilinear groups to
provide the anonymity of ciphertexts. Shi and Waters constructed a
delegatable HVE scheme \cite{ShiW08}. Lee and Lee constructed an efficient
HVE scheme with a constant number of pairing operations \cite{LeeL11}. In
composite order bilinear groups, the random blinding property using subgroups
provides the anonymity of ciphertexts and the orthogonal property among
subgroups provides the successful decryption. However, it is inefficient to
use composite order bilinear groups since the group order of composite order
bilinear groups should be large. To overcome this problem of inefficiency,
Freeman presented a general framework that converts cryptographic schemes
from composite order bilinear groups to prime order bilinear groups
\cite{Freeman10}. Ducas also showed that HVE schemes in composite order
bilinear groups are easily converted to schemes in prime order bilinear
groups \cite{Ducas10}. However, these conversion methods result in asymmetric
bilinear groups.

Finally, an HVE scheme can be derived from inner-product encryption (IPE)
\cite{KatzSW08,OkamotoT09,Park11}. IPE is a kind of PE and it enable the
evaluation of inner-product predicates between the vector of ciphertexts and
the vector of tokens. Katz et al. \cite{KatzSW08} constructed the first IPE
scheme under composite order bilinear groups. Okamoto and Takashima
constructed an hierarchical IPE scheme using dual pairing vector spaces
\cite{OkamotoT09}. Park proposed an IPE scheme under prime order bilinear
groups and proved its security under the well-known assumptions
\cite{Park11}. The main idea of converting an IPE scheme to an HVE scheme is
to construct a predicate of conjunctive equality using a predicate of inner
product \cite{KatzSW08}.

\section{Preliminaries}

In this section, we define hidden vector encryption, and introduce bilinear
groups of prime order and two complexity assumptions.

\subsection{Hidden Vector Encryption} \label{sec:back-hve}

Let $\Sigma$ be a finite set of attributes and let $*$ be a special symbol
not in $\Sigma$. Define $\Sigma_* = \Sigma \cup \{*\}$. The star $*$ plays
the role of a wild-card or ``don't care'' value. For a vector $\vec{\sigma} =
(\sigma_1, \ldots, \sigma_\ell) \in \Sigma_*^\ell$, we define a predicate
$f_{\vec{\sigma}}$ over $\Sigma^\ell$ as follows: For $\vec{x} = (x_1,
\ldots, x_\ell) \in \Sigma^\ell$, it set $f_{\vec{\sigma}}(\vec{x}) = 1$ if
$\forall i : (\sigma_i = x_i \mbox{ or } \sigma_i = *)$, it set
$f_{\vec{\sigma}}(\vec{x}) = 0$ otherwise.

\begin{definition}[Hidden Vector Encryption] \label{def:hve-syntax}
An HVE scheme consists of four algorithms \tb{Setup}, \tb{GenToken},
\tb{Encrypt}, and \tb{Query} which are defined as follows:
\begin{description}
\item \tb{Setup}($1^{\lambda}, \ell$): The setup algorithm takes as input a
    security parameter $1^{\lambda}$ and the length parameter $\ell$. It
    outputs a public key $PK$ and a secret key $SK$.

\item \tb{GenToken}($\vec{\sigma}, SK, PK$): The token generation algorithm
    takes as input a vector $\vec{\sigma} = (\sigma_1, \ldots, \sigma_\ell)
    \in \Sigma_{*}^\ell$ that corresponds to a predicate
    $f_{\vec{\sigma}}$, the secret key $SK$ and the public key $PK$. It
    outputs a token $TK_{\vec{\sigma}}$ for the vector $\vec{\sigma}$.

\item \tb{Encrypt}($\vec{x}, M, PK$): The encrypt algorithm takes as input
    a vector $\vec{x} = (x_1, \ldots, x_\ell) \in \Sigma^\ell$, a message
    $M \in \mathcal{M}$, and the public key $PK$. It outputs a ciphertext
    $CT$ for $\vec{x}$ and $M$.

\item \tb{Query}($CT, TK_{\vec{\sigma}}, PK$): The query algorithm takes as
    input a ciphertext $CT$, a token $TK_{\vec{\sigma}}$ for a vector
    $\vec{\sigma}$ that corresponds to a predicate $f_{\vec{\sigma}}$, and
    the public key $PK$. It outputs $M$ if $f_{\vec{\sigma}}(\vec{x})=1$ or
    outputs $\perp$ otherwise.
\end{description}
The scheme should satisfy the following correctness property: For all
$\vec{x} \in \Sigma^\ell$, $M \in \mathcal{M}$, $\vec{\sigma} \in
\Sigma_*^\ell$, let $(PK, SK) \leftarrow \tb{Setup}(1^{\lambda}, \ell)$, $CT
\leftarrow \tb{Encrypt}(\vec{x}, M, PK)$, and $TK_{\vec{\sigma}} \leftarrow
\tb{GenToken}(\sigma, SK, PK)$.
\begin{itemize}
\item If $f_{\vec{\sigma}}(\vec{x}) = 1$, then $\tb{Query}(CT,
    TK_{\vec{\sigma}}, PK) = M$.

\item If $f_{\vec{\sigma}}(\vec{x}) = 0$, then $\tb{Query}(CT,
    TK_{\vec{\sigma}}, PK) = \perp$ with all but negligible probability.
\end{itemize}
\end{definition}

\begin{definition}[Selective Security] \label{def:hve-seind}
The selective security of HVE is defined as the following game between a
challenger $\mc{C}$ and an adversary $\mc{A}$:
\begin{enumerate}
\item \tb{Init}: $\mc{A}$ submits two vectors $\vec{x}_0, \vec{x}_1 \in
    \Sigma^\ell$.

\item \tb{Setup}: $\mc{C}$ runs the setup algorithm and keeps the secret
    key $SK$ to itself, then it gives the public key $PK$ to $\mc{A}$.

\item \tb{Query 1}: $\mc{A}$ adaptively requests a polynomial number of
    tokens for vectors $\vec{\sigma}_1, \ldots, \vec{\sigma}_{q_1}$ that
    correspond to predicates $f_{\vec{\sigma}_1}, \ldots,
    f_{\vec{\sigma}_{q_1}}$ subject to the restriction that
    $f_{\vec{\sigma}_i}(\vec{x}_0) = f_{\vec{\sigma}_i}(\vec{x}_1)$ for all
    $i$. In responses, $\mc{C}$ gives the corresponding tokens
    $TK_{\vec{\sigma}_i}$ to $\mc{A}$.

\item \tb{Challenge}: $\mc{A}$ submits two messages $M_0, M_1$ subject to
    the restriction that if there is an index $i$ such that
    $f_{\vec{\sigma}_i}(\vec{x}_0) = f_{\vec{\sigma}_i}(\vec{x}_1) = 1$
    then $M_0 = M_1$. $\mc{C}$ chooses a random coin $\gamma$ and gives a
    ciphertext $CT$ of $(\vec{x}_{\gamma}, M_{\gamma})$ to $\mc{A}$.

\item \tb{Query 2}: $\mc{A}$ continues to request tokens for vectors
    $\vec{\sigma}_{q_1 +1}, \ldots, \vec{\sigma}_{q}$ that correspond to
    predicates $f_{\vec{\sigma}_{q_1 +1}}, \lb \ldots, f_{\vec{\sigma}_q}$
    subject to the two restrictions as before.

\item \tb{Guess}: $\mc{A}$ outputs a guess $\gamma'$. If $\gamma =
    \gamma'$, it outputs 0. Otherwise, it outputs 1.
\end{enumerate}
The advantage of $\mc{A}$ is defined as $\Adv_{\mc{A}}^{HVE}(1^\lambda) =
\big| \Pr[\gamma = \gamma'] - 1/2 \big|$ where the probability is taken over
the coin tosses made by $\mc{A}$ and $\mc{C}$. We say that an HVE scheme is
selectively secure if all probabilistic polynomial-time (PPT) adversaries
have at most a negligible advantage in the above game.
\end{definition}

\subsection{Bilinear Groups of Prime Order}

Let $\G$ and $\G_{T}$ be multiplicative cyclic groups of prime $p$ order. Let
$g$ be a generator of $\G$. The bilinear map $e:\G \times \G \rightarrow
\G_{T}$ has the following properties:
\begin{enumerate}
\item Bilinearity: $\forall u,v \in \G$ and $\forall a,b \in \Z_p$,
    $e(u^a,v^b) = e(u,v)^{ab}$.
\item Non-degeneracy: $\exists g$ such that $e(g,g)$ has order $p$, that
    is, $e(g,g)$ is a generator of $\G_T$.
\end{enumerate}
We say that $(p, \G, \G_T, e)$ are bilinear groups if the group operations in
$\G$ and $\G_T$ as well as the bilinear map $e$ are all efficiently
computable.

\subsection{Complexity Assumptions} \label{sec:assump-prime}

We introduce two simple assumptions under prime order bilinear groups. The
decisional bilinear Diffie-Hellman assumption was introduced in
\cite{BonehF01}. The parallel 3-party Diffie-Hellman (P3DH) assumption is
newly introduced in this paper.

\begin{assumption}[Decisional Bilinear Diffie-Hellman, DBDH]
Let $(p, \G, \G_T, e)$ be a description of the bilinear group of prime order
$p$. The DBDH problem is stated as follows: given a challenge tuple
    $$D = \big( (p, \G, \G_T, e),~
    g, g^a, g^b, g^c \big) \mbox{ and } T,$$
decides whether $T = T_0 = e(g, g)^{abc}$ or $T = T_1 = e(g,g)^d$ with random
choices of $a, b, c, d \in \Z_p$. The advantage of $\mc{A}$ is defined as
    $\Adv^{DBDH}_{\mc{A}} (1^\lambda) = \big|
    \Pr \big[\mc{A}(D, T_0) = 1 \big] -
    \Pr \big[\mc{A}(D, T_1) = 1 \big] \big|$
where the probability is taken over the random choices of $a, b, c, d \in
\Z_p$ and the random bits used by $\mc{A}$. We say that the DBDH assumption
holds if no PPT algorithm has a non-negligible advantage in solving the above
problem.
\end{assumption}

\begin{assumption}[Parallel 3-party Diffie-Hellman, P3DH]
Let $(p, \G, \G_T, e)$ be a description of the bilinear group of prime order
$p$. The P3DH problem is stated as follows: given a challenge tuple
    \begin{align*}
    D = \big(
    &   (p, \G, \G_T, e),~
        (g,f),(g^a, f^a), (g^b, f^b),
        (g^{ab} f^{z_1}, g^{z_1}), (g^{abc} f^{z_2}, g^{z_2})
        \big) \mbox{ and } T,
    \end{align*}
decides whether $T = T_0 = (g^{c} f^{z_3}, g^{z_3})$ or $T = T_1 = (g^d
f^{z_3}, g^{z_3})$ with random choices of $a, b, c, d \in \Z_p$ and $z_1,
z_2, z_3 \in \Z_p$. The advantage of $\mc{A}$ is defined as
    $\Adv^{P3DH}_{\mc{A}} (1^\lambda) = \big|
    \Pr \big[\mc{A}(D, T_0) = 1 \big] -
    \Pr \big[\mc{A}(D, T_1) = 1 \big] \big|$
where the probability is taken over the random choices of $a, b, c, d, z_1,
z_2, z_3$ and the random bits used by $\mc{A}$. We say that the P3DH
assumption holds if no PPT algorithm has a non-negligible advantage in
solving the above problem.
\end{assumption}

\begin{remark}
The P3DH problem can be modified as follows: given a challenge tuple
    $D = \big( (p, \G, \G_T, e),~
    (g,f), \lb (g^a, f^a), (g^b, f^b),
    (g^{ab} f^{z_1}, g^{z_1}), (g^{c} f^{z_2}, g^{z_2}) \big)$
    and $T$,
decides whether $T = T_0 = (g^{abc} f^{z_3}, g^{z_3})$ or $T = T_1 = (g^d
f^{z_3}, g^{z_3})$. However, this modified one is the same as the original
one by changing the position of the challenge tuple as
    $D = \big( (p, \G, \G_T, e),~
    (g,f), (g^a, f^a), (g^b, f^b),
    (g^{ab} f^{z_1}, g^{z_1}), T \big)$
    and $T' = (g^{c} f^{z_2}, g^{z_2})$,
Thus, we will use any one of challenge tuple forms for the P3DH assumption.
\end{remark}

\section{Our Techniques}

The basic idea to convert HVE schemes from composite order bilinear groups to
prime order bilinear groups is to use bilinear product groups that are
extended from bilinear groups using the direct product operation. Bilinear
product groups were widely used in dual system encryption of Waters
\cite{Waters09,LewkoW10}, private linear broadcast encryption of Garg et al.
\cite{GargKSW10}, and the conversion method of Freeman \cite{Freeman10}. The
product groups extended from multiplicative cyclic groups represent an
exponent as a vector. Thus vector operations in product groups and bilinear
product groups should be defined. Definition \ref{def:vector-ops} and
Definition \ref{def:bilinear-prod} define the vector operations in product
groups and bilinear product groups, respectively.

\begin{definition}[Vector Operations] \label{def:vector-ops}
Let $\G$ be multiplicative cyclic groups of prime $p$ order. Let $g$ be a
generator of $\G$. We define vector operations over $\G$ as follows:
\begin{enumerate}
\item For a vector $\vec{b} = (b_1, \ldots, b_n) \in \Z_p^n$, define
    $g^{\vec{b}} := (g^{b_1}, \ldots, g^{b_n}) \in \G^n$.

\item For a vector $\vec{b} = (b_1, \ldots, b_n) \in \Z_p^n$ and a scalar
    $c \in \Z_p$, define $(g^{\vec{b}})^c := (g^{b_1 c}, \ldots, g^{b_n c})
    \in \G^n$.

\item For two vectors $\vec{a} = (a_1, \ldots, a_n), \vec{b} = (b_1,
    \ldots, b_n) \in \Z_p^n$, define $g^{\vec{a}} g^{\vec{b}} := (g^{a_1 +
    b_1}, \ldots, g^{a_n + b_n}) \in \G^n$.
\end{enumerate}
\end{definition}

\begin{definition}[Bilinear Product Groups] \label{def:bilinear-prod}
Let $(p, \G, \G_{T}, e)$ be bilinear groups of prime order. Let $g$ be a
generator of $\G$. For integers $n$ and $m$, the bilinear product groups
$((p, \G, \G_T, e), g^{\vec{b}_1}, \ldots, g^{\vec{b}_m})$ of basis vectors
$\vec{b}_1, \ldots, \vec{b}_m$ is defined as follows
\begin{enumerate}
\item The basis vectors $\vec{b}_1, \ldots, \vec{b}_m$ are random vectors
    such that $\vec{b}_i = (b_{i,1}, \ldots, b_{i,n}) \in \Z_p^n$.

\item The bilinear map $e:\G^n \times \G^n \rightarrow \G_T$ is defined as
    $e(g^{\vec{a}},g^{\vec{b}}) := \prod_{i=1}^n e(g^{a_i}, g^{b_i}) =
    e(g,g)^{\vec{a} \cdot \vec{b}}$ where $\cdot$ is the inner product
    operation.
\end{enumerate}
\end{definition}

To guarantee the correctness of cryptographic schemes in bilinear product
groups, the orthogonal property of composite order bilinear groups should be
implemented in bilinear product groups. The previous research \cite{Waters09,
GargKSW10,Freeman10,LewkoW10} showed that the orthogonal property can be
implemented in bilinear product groups. The idea is that the orthogonality
between vectors can be defined using the inner-product operation such that
$\vec{x} \cdot \vec{y} = 0$ since the bilinear map provides the inner-product
operation. Definition \ref{def:orth} define the orthogonality in bilinear
product groups.

\begin{definition}[Orthogonality] \label{def:orth}
Let $((p, \G, \G_{T}, e), g^{\vec{b}_1}, \ldots, g^{\vec{b}_m})$ be bilinear
product groups with $n,m$ parameters. Let $G_i, G_j$ be subgroups spanned by
$g^{\vec{b}_i}, g^{\vec{b}_j}$, respectively. That is, $G_i = \langle
g^{\vec{b}_i} \rangle$ and $G_j = \langle g^{\vec{b}_j} \rangle$. Then the
two subgroups $G_i$ and $G_j$ are orthogonal to each other if $e(\vec{A},
\vec{B}) = 1$ for all $\vec{A} \in G_i$ and $\vec{B} \in G_j$.
\end{definition}

The main idea of our method that converts HVE schemes from composite order
bilinear groups to prime order bilinear groups is that the previous HVE
schemes \cite{BonehW07,ShiW08,LeeL11} in composite order bilinear groups use
the composite 3-party Diffie-Hellman (C3DH) assumption that is not a kind of
the subgroup decision (SD) assumption.

The SD assumption is to distinguish whether $h \in \G$ or $h \in \G_1$ where
$\G$ is a group and $\G_1$ is a subgroup of $\G$ \cite{BonehGN05}. In product
groups $\G^n$, a subgroup $G$ is defined as a vector space spanned by some
basis vectors $\vec{b}_1, \ldots, \vec{b}_m$ such that $G = \langle
g^{\vec{b}_1}, \ldots, g^{\vec{b}_m} \rangle$. If a subgroup is constructed
from one basis vector, then the SD assumption is related to the DDH
assumption. If a subgroup is constructed from $k$ number of basis vectors,
then the SD assumption is related to the decisional $k$-Linear ($k$-DLIN)
assumption \cite{Freeman10}. In symmetric bilinear groups of prime order, a
subgroup should be constructed from two basis vectors since the DDH
assumption is not valid \cite{Waters09,GargKSW10}. If a subgroup is
constructed from two basis vectors, then cryptographic schemes become
complicated and there is no generic conversion method from composite order
groups to prime order groups. In asymmetric bilinear groups of prime order, a
subgroup can be constructed from one basis vector since the DDH assumption is
valid \cite{Freeman10,LewkoW10}. If a subgroup is constructed from one basis
vector, then there is a generic conversion method of Freeman, but it only
works in asymmetric bilinear groups.

The C3DH assumption is defined in Assumption \ref{assump:c3dh}. The notable
properties of the C3DH assumption are that the target value $T$ is always an
element of $\G_{p_1 p_2}$ in contrast to the SD assumption, and the subgroup
$\G_{p_2}$ plays the role of random blinding. From these properties of the
C3DH assumption, it is possible to use just one basis vector to construct a
subgroup. Additionally, it is possible to use simple basis vectors for
cryptographic schemes since ciphertexts and tokens can use different
subgroups that are not orthogonal.

\begin{assumption}[Composite 3-party Diffie-Hellman, C3DH] \label{assump:c3dh}
Let $(N, \G, \G_{T}, e)$ be a description of bilinear groups of composite
order $N = p_1 \cdots p_m$ where $p_i$ is a random prime. Let $g_{p_i}$ be a
generator of the subgroup $\G_{p_i}$. The C3DH assumption is stated as
follows: given a challenge tuple
    $$\vec{D} = \big( (N, \G, \G_T, e),
    g_{p_1}, \ldots, g_{p_m},
    g_{p_1}^a, g_{p_1}^b, g_{p_1}^{ab} R_1, g_{p_1}^{abc} R_2 \big)
    \mbox{ and } T,$$
decides whether $T = T_0 = g_{p_1}^{c} R_3$ or $T = T_1 = g_{p_1}^d R_3$ with
random choices of $a,b,c,d \in \Z_{p_1}$ and $R_1,R_2,R_3 \in \G_{p_2}$.
\end{assumption}

For instance, we select basis vectors $\vec{b}_{1,1} = (1,0), \vec{b}_{1,2} =
(1,a), \vec{b}_{2} = (a,-1)$ for the conversion from bilinear groups of
composite $N = p_1 p_2$ order. For the conversion from bilinear groups of
composite $N = p_1 p_2 p_3$ order, we select basis vectors $\vec{b}_{1,1} =
(1,0,a_1), \vec{b}_{1,2} = (1,a_2,0), \vec{b}_{2} = (a_2,-1,a_1a_2-a_3),
\vec{b}_3 = (a_1,a_3,-1)$. Although different basis vectors were selected,
the assumption for the security proof is the simple one that is independent
of basis vectors.

\section{Conversion 1: BW-HVE} \label{sec:bw-hve}

In this section, we convert the HVE scheme of Boneh and Waters
\cite{BonehW07} in composite order bilinear groups to an HVE scheme in prime
order bilinear groups and prove its selective security under the DBDH and
P3DH assumptions.

\subsection{Construction}

\begin{description}
\item [\tb{Setup}($1^{\lambda}, \ell$):] It first generates the bilinear
    group $\G$ of prime order $p$ of bit size $\Theta(\lambda)$. It chooses
    a random value $a \in \Z_p$ and sets basis vectors for bilinear product
    groups as $\vec{b}_{1,1} = (1, 0),~ \vec{b}_{1,2} = (1, a),~
    \vec{b}_{2} = (a, -1)$. It also sets $\vec{B}_{1,1} =
    g^{\vec{b}_{1,1}}, \vec{B}_{1,2} = g^{\vec{b}_{1,2}}, \vec{B}_2 =
    g^{\vec{b}_2}$.
    It selects random exponents $v', \{ u'_i, h'_i, w'_i \}_{i=1}^\ell,
    \alpha \in \Z_p$, $z_v, \{ z_{u,i}, z_{h,i}, z_{w,i} \}_{i=1}^\ell \in
    \Z_p$ and outputs a secret key and a public key as
    \begin{align*}
    SK = \Big(~
    &   \vec{V}_k = \vec{B}_{1,2}^{v'},
        \big\{
        {\vec{U}_{k,i}} = \vec{B}_{1,2}^{u'_i},~
        {\vec{H}_{k,i}} = \vec{B}_{1,2}^{h'_i},~
        {\vec{W}_{k,i}} = \vec{B}_{1,2}^{w'_i}
        \big\}_{i=1}^\ell,
        \vec{B}_{1,2}^{\alpha}
    ~\Big), \\
    PK = \Big(~
    &   \vec{B}_{1,1},~ \vec{B}_{1,2},~ \vec{B}_2,~
        \vec{V}_c = \vec{B}_{1,1}^{v'} \vec{B}_2^{z_v},~
        \big\{
        \vec{U}_{c,i} = \vec{B}_{1,1}^{u'_i} \vec{B}_2^{z_{u,i}},~
        \vec{H}_{c,i} = \vec{B}_{1,1}^{h'_i} \vec{B}_2^{z_{h,i}},~
        \vec{W}_{c,i} = \vec{B}_{1,1}^{w'_i} \vec{B}_2^{z_{w,i}}
        \big\}_{i=1}^\ell,~ \\
    &   \Omega = e(\vec{B}_{1,1}^{v'}, \vec{B}_{1,2})^{\alpha}
    ~\Big).
    \end{align*}

\item [\tb{GenToken}($\vec{\sigma}, SK, PK$):] It takes as input a vector
    $\vec{\sigma} = (\sigma_1, \ldots, \sigma_\ell) \in \Sigma_*^\ell$, the
    secret key $SK$, and the public key $PK$. Let $S$ be the set of indexes
    that are not wild-card fields in the vector $\vec{\sigma}$. It selects
    random exponents $\{ r_{1,i}, r_{2,i} \}_{i \in S} \in \Z_p$ and
    outputs a token as
    \begin{align*}
    TK_{\vec{\sigma}} = \Big(~
    &   \vec{K}_1 = \vec{B}_{1,2}^{\alpha} \prod_{i \in S} \big(
                    \vec{U}_{k,i}^{\sigma_i} \vec{H}_{k,i} \big)^{r_{1,i}}
                    \vec{W}_{k,i}^{r_{2,i}},~
        \big\{
        \vec{K}_{2,i} = \vec{V}_k^{-r_{1,i}},~
        \vec{K}_{3,i} = \vec{V}_k^{-r_{2,i}}
        \big\}_{i \in S}
    ~\Big).
    \end{align*}

\item [\tb{Encrypt}($\vec{x}, M, PK$):] It takes as input a vector $\vec{x}
    = (x_1, \ldots, x_\ell) \in \Sigma^\ell$, a message $M \in \mc{M}$, and
    the public key $PK$. It first chooses a random exponent $t \in \Z_p$
    and random blinding values $z_1, \{ z_{2,i}, z_{3,i} \}_{i=1}^\ell \in
    \Z_p$. Then it outputs a ciphertext as
    \begin{align*}
    CT = \Big(~
        C_0 = \Omega^t M,~
        \vec{C}_1 = \vec{V}_c^t \vec{B}_2^{z_1},~
        \big\{
        \vec{C}_{2,i} = (\vec{U}_{c,i}^{x_i} \vec{H}_{c,i})^t \vec{B}_2^{z_{2,i}},~
        \vec{C}_{3,i} = \vec{W}_{c,i}^t \vec{B}_2^{z_{3,i}}
        \big\}_{i=1}^\ell
    ~\Big).
    \end{align*}

\item [\tb{Query}($CT, TK_{\vec{\sigma}}, PK$):] It takes as input a
    ciphertext $CT$ and a token $TK_{\vec{\sigma}}$ of a vector
    $\vec{\sigma}$. It first computes
    \begin{align*}
    M \leftarrow C_0 \cdot \Big( e(\vec{C}_1, \vec{K}_1) \cdot
        \prod_{i \in S} e(\vec{C}_{2,i}, \vec{K}_{2,i}) \cdot
            e(\vec{C}_{3,i}, \vec{K}_{3,i}) \Big)^{-1}.
    \end{align*}
    If $M \notin \mathcal{M}$, it outputs $\perp$ indicating that the
    predicate $f_{\vec{\sigma}}$ is not satisfied. Otherwise, it outputs
    $M$ indicating that the predicate $f_{\vec{\sigma}}$ is satisfied.
\end{description}

\subsection{Correctness}

If $f_{\vec{\sigma}}(\vec{x}) = 1$, then the following calculation shows that
$\textbf{Query}(CT, TK_{\vec{\sigma}}, PK) = M$ using the orthogonality of
basis vectors such that $e(g^{\vec{b}_2}, g^{\vec{b}_{1,2}})=1$.
    \begin{align*}
    \lefteqn{ e(\vec{C}_1, \vec{K}_1) \cdot
        \prod_{i \in S} \big( e(\vec{C}_{2,i}, \vec{K}_{2,i}) \cdot
            e(\vec{C}_{3,i}, \vec{K}_{3,i}) \big) } \\
    &=  e(\vec{V}_c^{t}, \vec{B}_{1,2}^{\alpha} \prod_{i \in S} (
                \vec{U}_{k,i}^{\sigma_i} \vec{H}_{k,i})^{r_{1,i}}
                \vec{W}_{k,i}^{r_{2,i}}) \cdot
        \prod_{i \in S}
            e((\vec{U}_{c,i}^{x_i} \vec{H}_{c,i})^t, \vec{V}_k^{-r_{1,i}}) \cdot
            e(\vec{W}_{c,i}^t, \vec{V}_k^{-r_{2,i}}) \\
    &=  e(\vec{B}_{1,1}^{v' t}, \vec{B}_{1,2}^{\alpha}) \cdot
        \prod_{i \in S} e(g^{v'}, g^{u'_i (\sigma_i - x_i)})^{t \cdot r_{1,i}}
     =  e(\vec{B}_{1,1}^{v'}, \vec{B}_{1,2})^{\alpha t}.
    \end{align*}
Otherwise, that is $f_{\vec{\sigma}}(\vec{x}) = 0$, then the probability of
$\textbf{Query}(CT, TK_{\vec{\sigma}}, PK) \neq \perp$ is negligible by
limiting $|\mathcal{M}|$ to less than $|\G_T|^{1/4}$.

\subsection{Security}

\begin{theorem} \label{thm:bw-hve}
The above HVE scheme is selectively secure under the DBDH and P3DH
assumptions.
\end{theorem}

\begin{proof}
The proof of this theorem is easily obtained from the following four Lemmas
\ref{lem:bw-hve}, \ref{lem:bw-hve-a1}, \ref{lem:bw-hve-a2}, and
\ref{lem:bw-hve-a3}. Before presenting the four lemmas, we first introduce
the following three assumptions. The HVE scheme of Boneh and Waters
constructed in bilinear groups of composite order $N = p_1 p_2$, and its
security was proven under the DBDH, bilinear subgroup decision (BSD), and
C3DH assumptions \cite{BonehW07}. These assumptions in composite order
bilinear groups are converted to the following Assumptions
\ref*{sec:bw-hve}-1, \ref*{sec:bw-hve}-2, and \ref*{sec:bw-hve}-3 using our
conversion method.
\end{proof}

\vs \noindent \textbf{Assumption \ref*{sec:bw-hve}-1} Let $((p, \G, \G_T, e),
g^{\vec{b}_{1,1}}, g^{\vec{b}_{1,2}}, g^{\vec{b}_{2}})$ be the bilinear
product group of basis vectors $\vec{b}_{1,1} = (1,0), \vec{b}_{1,2} = (1,a),
\vec{b}_{2} = (a,-1)$. The Assumption \ref*{sec:bw-hve}-1 is stated as
follows: given a challenge tuple
    $$D = \big( (p, \G, \G_T, e),~
    g^{\vec{b}_{1,1}}, g^{\vec{b}_{1,2}}, g^{\vec{b}_2},
    (g^{\vec{b}_{1,1}})^{c_1}, (g^{\vec{b}_{1,1}})^{c_2},
    (g^{\vec{b}_{1,2}})^{c_1}, (g^{\vec{b}_{1,2}})^{c_2},
    (g^{\vec{b}_{1,1}})^{c_3} \big) \mbox{ and } T,$$
decides whether $T = T_0 = e(g, g)^{c_1 c_2 c_3}$ or $T = T_1 = e(g,g)^d$
with random choices of $c_1, c_2, c_3, d \in \Z_p$.

\vs \noindent \textbf{Assumption \ref*{sec:bw-hve}-2} Let $((p, \G, \G_T, e),
g^{\vec{b}_{1,1}}, g^{\vec{b}_{1,2}}, g^{\vec{b}_{2}})$ be the bilinear
product group of basis vectors $\vec{b}_{1,1} = (1,0), \vec{b}_{1,2} = (1,a),
\vec{b}_{2} = (a,-1)$. The Assumption \ref*{sec:bw-hve}-2 is stated as
follows: given a challenge tuple
    $$D = \big( (p, \G, \G_T, e),~
    g^{\vec{b}_{1,1}}, g^{\vec{b}_{1,2}}, g^{\vec{b}_2}
    \big) \mbox{ and } T,$$
decides whether $T = T_0 = e((g^{\vec{b}_{1,1}})^{c_1} (g^{\vec{b}_2})^{c_3},
(g^{\vec{b}_{1,2}})^{c_2})$ or $T = T_1 = e((g^{\vec{b}_{1,1}})^{c_1},
(g^{\vec{b}_{1,2}})^{c_2})$ with random choices of $c_1, c_2, c_3 \in \Z_p$.

\vs \noindent \textbf{Assumption \ref*{sec:bw-hve}-3} Let $((p, \G, \G_T, e),
g^{\vec{b}_{1,1}}, g^{\vec{b}_{1,2}}, g^{\vec{b}_{2}})$ be the bilinear
product group of basis vectors $\vec{b}_{1,1} = (1,0), \vec{b}_{1,2} = (1,a),
\vec{b}_{2} = (a,-1)$. The Assumption \ref*{sec:bw-hve}-3 is stated as
follows: given a challenge tuple
    \begin{eqnarray*}
    D = \big( (p, \G, \G_T, e),~
    g^{\vec{b}_{1,1}}, g^{\vec{b}_{1,2}}, g^{\vec{b}_2},
    (g^{\vec{b}_{1,2}})^{c_1}, (g^{\vec{b}_{1,2}})^{c_2},
    (g^{\vec{b}_{1,1}})^{c_1 c_2} (g^{\vec{b}_2})^{z_1},
    (g^{\vec{b}_{1,1}})^{c_1 c_2 c_3} (g^{\vec{b}_2})^{z_2}
    \big) \mbox{ and } \vec{T},
    \end{eqnarray*}
decides whether $T = T_0 = (g^{\vec{b}_{1,1}})^{c_3} (g^{\vec{b}_2})^{z_3}$
or $T = T_1 = (g^{\vec{b}_{1,1}})^d (g^{\vec{b}_2})^{z_3}$ with random
choices of $c_1, c_2, c_3, d \in \Z_p$ and $z_1, z_2, z_3 \in \Z_p$.

\begin{lemma} \label{lem:bw-hve}
The above HVE scheme is selectively secure under the Assumptions
\ref*{sec:bw-hve}-1, \ref*{sec:bw-hve}-2, and \ref*{sec:bw-hve}-3.
\end{lemma}

\begin{proof}
The proof of this lemma is directly obtained from \cite{BonehW07} since the
Assumptions \ref*{sec:bw-hve}-1, \ref*{sec:bw-hve}-2, and \ref*{sec:bw-hve}-2
in prime order bilinear groups are correspond to the DBDH, BSD, and C3DH
assumptions in composite order bilinear groups. That is, the proof of
\cite{BonehW07} can be exactly simulated using the vector operations in the
Definition \ref{def:vector-ops} and the Assumptions \ref*{sec:bw-hve}-1,
\ref*{sec:bw-hve}-2, and \ref*{sec:bw-hve}-3.
\end{proof}

\begin{lemma} \label{lem:bw-hve-a1}
If the DBDH assumption holds, then the Assumption \ref*{sec:bw-hve}-1 also
holds.
\end{lemma}

\begin{proof}
Suppose there exists an adversary $\mc{A}$ that breaks the Assumption
\ref*{sec:bw-hve}-1 with a non-negligible advantage. An algorithm $\mc{B}$
that solves the DBDH assumption using $\mc{A}$ is given: a challenge tuple $D
= ((p, \G, \G_T, e), g, g^{c_1}, g^{c_2}, g^{c_3})$ and $T$ where $T = T_0 =
e(g,g)^{c_1 c_2 c_3}$ or $T = T_1 = e(g,g)^d$. $\mc{B}$ first chooses random
values $a \in \Z_p$ and computes
    \begin{align*}
    &   g^{\vec{b}_{1,1}} = (g, 1),~
        g^{\vec{b}_{1,2}} = (g, g^a),~
        g^{\vec{b}_2} = (g^a, g^{-1}),~ \\
    &   (g^{\vec{b}_{1,1}})^{c_1} = (g^{c_1}, 1),~
        (g^{\vec{b}_{1,1}})^{c_2} = (g^{c_2}, 1),~
        (g^{\vec{b}_{1,1}})^{c_3} = (g^{c_3}, 1),~ \\
    &   (g^{\vec{b}_{1,2}})^{c_1} = (g^{c_1}, (g^{c_1})^a),~
        (g^{\vec{b}_{1,2}})^{c_2} = (g^{c_2}, (g^{c_2})^a).
    \end{align*}
Next, it gives the tuple
    $D' = ((p,\G, \G_T, e),
    g^{\vec{b}_{1,1}}, g^{\vec{b}_{1,2}}, g^{\vec{b}_2},
    (g^{\vec{b}_{1,1}})^{c_1}, (g^{\vec{b}_{1,1}})^{c_2},
    (g^{\vec{b}_{1,2}})^{c_1}, (g^{\vec{b}_{1,2}})^{c_2},
    (g^{\vec{b}_{1,1}})^{c_3})$ and $T$
to $\mc{A}$. Then $\mc{A}$ outputs a guess $\gamma'$. $\mc{B}$ also outputs
$\gamma'$. If the advantage of $\mc{A}$ is $\epsilon$, then the advantage of
$\mc{B}$ is greater than $\epsilon$ since the distribution of the challenge
tuple to $\mc{A}$ is equal to the Assumption \ref*{sec:bw-hve}-1.
\end{proof}

\begin{lemma} \label{lem:bw-hve-a2}
The Assumption \ref*{sec:bw-hve}-2 holds for all adversaries.
\end{lemma}

\begin{proof}
The equation $e((g^{\vec{b}_{1,1}})^{c_1} (g^{\vec{b}_2})^{c_3},
(g^{\vec{b}_{1,2}})^{c_2}) = e((g^{\vec{b}_{1,1}})^{c_1},
(g^{\vec{b}_{1,2}})^{c_2})$ holds by the orthogonality of basis vectors such
that $e(g^{\vec{b}_{2}}, g^{\vec{b}_{1,2}}) = 1$. Therefore, any adversary
can not break the Assumption \ref*{sec:bw-hve}-2.
\end{proof}

\begin{lemma} \label{lem:bw-hve-a3}
If the P3DH assumption holds, then the Assumption \ref*{sec:bw-hve}-3 also
holds.
\end{lemma}

\begin{proof}
Suppose there exists an adversary $\mc{A}$ that breaks the Assumption
\ref*{sec:bw-hve}-3 with a non-negligible advantage. An algorithm $\mc{B}$
that solves the P3DH assumption using $\mc{A}$ is given: a challenge tuple
    $D = ((p, \G, \G_T, e), (g,f),
    (g^{c_1}, f^{c_1}), (g^{c_2}, f^{c_2}),
    (g^{c_1 c_2} f^{z_1}, g^{z_1}), (g^{c_1 c_2 c_3} f^{z_2}, g^{z_2}))$
and $T$ where $T = T_0 = (g^{c_3} f^{z_3}, g^{z_3})$ or $T = T_1 = (g^d
f^{z_3}, g^{z_3})$. $\mc{B}$ first computes
    \begin{align*}
    &   g^{\vec{b}_{1,1}} = (g, 1),~
        g^{\vec{b}_{1,2}} = (g, f),~
        g^{\vec{b}_2}     = (f, g^{-1}),~ \\
    &   (g^{\vec{b}_{1,2}})^{c_1} = (g^{c_1}, f^{c_1}),~
        (g^{\vec{b}_{1,2}})^{c_2} = (g^{c_2}, f^{c_2}),~ \\
    &   (g^{\vec{b}_{1,1}})^{c_1 c_2} (g^{\vec{b}_2})^{z_1}
            = (g^{c_1 c_2} f^{z_1}, (g^{z_1})^{-1}),~
        (g^{\vec{b}_{1,1}})^{c_1 c_2 c_3} (g^{\vec{b}_2})^{z_2}
            = (g^{c_1 c_2 c_3} f^{z_2}, (g^{z_2})^{-1}).
    \end{align*}
Intuitively, it sets $a = \log f$. Next, it gives the tuple
    $D' = ((p,\G, \G_T, e),
    g^{\vec{b}_{1,1}}, g^{\vec{b}_{1,2}}, g^{\vec{b}_2},
    (g^{\vec{b}_{1,1}})^{c_1}, (g^{\vec{b}_{1,1}})^{c_2}, \lb
    (g^{\vec{b}_{1,2}})^{c_1}, (g^{\vec{b}_{1,2}})^{c_2},
    (g^{\vec{b}_{1,1}})^{c_1 c_2 c_3})$ and $T$
to $\mc{A}$. Then $\mc{A}$ outputs a guess $\gamma'$. $\mc{B}$ also outputs
$\gamma'$. If the advantage of $\mc{A}$ is $\epsilon$, then the advantage of
$\mc{B}$ is greater than $\epsilon$ since the distribution of the challenge
tuple to $\mc{A}$ is equal to the Assumption \ref*{sec:bw-hve}-3.
\end{proof}

\section{Conversion 2: SW-dHVE} \label{sec:sw-dhve}

In this section, we convert the delegatable HVE scheme of Shi and Waters
\cite{ShiW08} to prime order bilinear groups and prove its selective security
under the DBDH and P3DH assumptions.

\subsection{Construction}

Let $\Sigma$ be a finite set of attributes and let $?, *$ be two special
symbol not in $\Sigma$. Define $\Sigma_{?,*} = \Sigma \cup \{?,*\}$. The
symbol $?$ denotes a delegatable field, i.e., a field where one is allowed to
fill in an arbitrary value and perform delegation. The symbol $*$ denotes a
wild-card field or ``don't care'' field.

\begin{description}
\item [\tb{Setup}($1^{\lambda}, \ell$):] It first generates the bilinear
    group $\G$ of prime order $p$ of bit size $\Theta(\lambda)$. It chooses
    random values $a_1, a_2, a_3 \in \Z_p$ and sets basis vectors for
    bilinear product groups as $\vec{b}_{1,1} = (1, 0, a_1),~ \vec{b}_{1,2}
    = (1, a_2, 0),~ \vec{b}_2 = (a_2, -1, a_1 a_2 - a_3),~ \vec{b}_3 =
    (a_1, a_3, -1)$.
    It also sets
    \begin{align*}
    &   \vec{B}_{1,1} = g^{\vec{b}_{1,1}},~
        \vec{B}_{1,2} = g^{\vec{b}_{1,2}},~
        \vec{B}_2 = g^{\vec{b}_2},~
        \vec{B}_3 = g^{\vec{b}_3}.
    \end{align*}
    It selects random exponents $v', w'_1, w'_2, \{ u'_i, h'_i
    \}_{i=1}^\ell, \alpha \in \Z_p$, $z_v, z_{w,1}, z_{w,2}, \{ z_{u,i},
    z_{h,i} \}_{i=1}^\ell \in \Z_p$ and outputs a secret key and a public
    key as
    \begin{align*}
    SK = \Big(~
    &   \vec{V}_k = \vec{B}_{1,2}^{v'},
        \vec{W}_{k,1} = \vec{B}_{1,2}^{w'_1},
        \vec{W}_{k,2} = \vec{B}_{1,2}^{w'_2},~
        \big\{
        \vec{U}_{k,i} = \vec{B}_{1,2}^{u'_i},
        \vec{H}_{k,i} = \vec{B}_{1,2}^{h'_i}
        \big\}_{i=1}^\ell,~
        \vec{B}_{1,2}^{\alpha}
        \Big), \db \\
    PK = \Big(~
    &   \vec{B}_{1,1},~ \vec{B}_{1,2},~ \vec{B}_2,~ \vec{B}_3,~
        \vec{V}_{c} = \vec{B}_{1,1}^{v'} \vec{B}_2^{z_v},~
        \vec{W}_{c,1} = \vec{B}_{1,1}^{w'_1} \vec{B}_2^{z_{w,1}},~
        \vec{W}_{c,2} = \vec{B}_{1,2}^{w'_2} \vec{B}_2^{z_{w,2}},~ \\
    &   \big\{
        \vec{U}_{c,i} = \vec{B}_{1,1}^{u'_i} \vec{B}_2^{z_{u,i}},~
        \vec{H}_{c,i} = \vec{B}_{1,1}^{h'_i} \vec{B}_2^{z_{h,i}}
        \big\}_{i=1}^\ell,~
        \Omega = e(\vec{B}_{1,1}^{v'}, \vec{B}_{1,2})^{\alpha} ~\Big).
    \end{align*}

\item [\tb{GenToken}($\vec{\sigma}, SK, PK$):] It takes as input an
    attribute vector $\vec{\sigma} = (\sigma_1, \ldots, \sigma_\ell) \in
    \Sigma_{?,*}^\ell$ and the secret key $SK$.
    \begin{enumerate}
    \item Let $S$ be the set of indexes that are not delegatable fields
        and wild-card fields in the vector $\vec{\sigma}$. It first
        selects random exponents $r_1, r_2, \{ r_{3,i} \}_{i \in S} \in
        \Z_p$ and random blinding values $y_1, y_2, y_3, \{ y_{4,i} \}_{i
        \in S} \in \Z_p$. Then it computes decryption components as
        \begin{align*}
        &   \vec{K}_1 = \vec{B}_{1,2}^{\alpha}
            \vec{W}_{k,1}^{r_1} \vec{W}_{k,2}^{r_2}
            \prod_{i \in S} (\vec{U}_{k,i}^{\sigma_i} \vec{H}_{k,i})^{r_{3,i}}
            \vec{B}_3^{y_1},~
            \vec{K}_2 = \vec{V}_k^{-r_1} \vec{B}_3^{y_2},~
            \vec{K}_3 = \vec{V}_k^{-r_2} \vec{B}_3^{y_3},~ \\
        &   \big\{
            \vec{K}_{4,i} = \vec{V}_k^{-r_{3,i}} \vec{B}_3^{y_{4,i}}
            \big\}_{i \in S}.
        \end{align*}

    \item Let $S_?$ be the set of indexes that are delegatable fields. It
        selects random exponents $\{ s_{1,j}, s_{2,j}, \lb \{ s_{3,j,i}
        \} \} \in \Z_p$ and random blinding values $\{ y_{1,j,u},
        y_{1,j,h}, y_{2,j}, y_{3,j}, \{ y_{4,j,i} \} \} \in \Z_p$. Next,
        it computes delegation components as
        \begin{align*}
        \forall j \in S_? :~
        &   \vec{L}_{1,j,u} = \vec{U}_{k,i}^{s_{3,j,j}} \vec{B}_3^{y_{1,j,u}},~
            \vec{L}_{1,j,h} =
                \vec{W}_{k,1}^{s_{1,j}} \vec{W}_{k,2}^{s_{2,j}}
                \prod_{i \in S} (\vec{U}_{k,i}^{\sigma_i} \vec{H}_{k,i})^{s_{3,j,i}}
                \vec{H}_{k,j}^{s_{3,j,j}} \vec{B}_3^{y_{1,j,h}},~ \\
        &   \vec{L}_{2,j} = \vec{V}_k^{-s_{1,j}} \vec{B}_3^{y_{2,j}},~
            \vec{L}_{3,j} = \vec{V}_k^{-s_{2,j}} \vec{B}_3^{y_{3,j}},~
            \big\{
            \vec{L}_{4,j,i} = \vec{V}_k^{-s_{3,j,i}} \vec{B}_3^{y_{4,j,i}}
            \big\}_{i \in S \cup \{j\}}.
        \end{align*}

    \item Finally, it outputs a token as
        \begin{align*}
        TK_{\vec{\sigma}} = \Big(
            \vec{K}_1,~ \vec{K}_2,~ \vec{K}_3,~ \{ \vec{K}_{4,i} \}_{i \in S},~
            \big\{ \vec{L}_{1,j,u},~ \vec{L}_{1,j,h},~ \vec{L}_{2,j},~
                \vec{L}_{3,j},~ \{ \vec{L}_{4,j,i} \}_{i \in S \cup \{j\}}
            \big\}_{j \in S_?}
        \Big).
        \end{align*}
    \end{enumerate}

\item [\tb{Delegate}($\vec{\sigma}', TK_{\vec{\sigma}}, PK$):] It takes as
    input an attribute vector $\vec{\sigma}' = (\sigma_1, \ldots,
    \sigma_\ell) \in \Sigma_{?,*}^\ell$ and a token $TK_{\vec{\sigma}}$.
    Without loss of generality, we assume that $\sigma'$ fixes only one
    delegatable field of $\sigma$. It is clear that we can perform
    delegation on multiple fields if we have an algorithm to perform
    delegation on one field. Suppose $\sigma'$ fixes the $k$-th index of
    $\sigma$.
    \begin{enumerate}
    \item If the $k$-th index of $\sigma'$ is set to $*$, that is, a
        wild-card field, then it can perform delegation by simply
        removing the delegation components that correspond to $k$-th
        index.

    \item Otherwise, that is, if the $k$-th index of $\sigma'$ is set to
        some value in $\Sigma$, then it perform delegation as follows:
        \begin{enumerate}
        \item Let $S$ be the set of indexes that are not delegatable fields
            and wild-card fields in the vector $\vec{\sigma'}$. Note that
            $k \in S$. It selects random exponents $\mu, y_1, y_2, y_3, \{
            y_{4,i} \}_{i \in S} \in \Z_p$ and updates the token as
            \begin{align*}
            &   \vec{K}'_1 = \vec{K}_1
                (\vec{L}_{1,k,u}^{\sigma_k} \vec{L}_{1,k,h})^{\mu} \vec{B}_3^{y_1},~
                \vec{K}'_2 = \vec{K}_2 \vec{L}_{2,k}^{\mu} \vec{B}_3^{y_2},~
                \vec{K}'_3 = \vec{K}_3 \vec{L}_{3,k}^{\mu} \vec{B}_3^{y_3},~ \\
            &   \vec{K}'_{4,k} = \vec{L}_{4,k,k}^{\mu} \vec{B}_3^{y_{4,k}},~
                \big\{
                \vec{K}'_{4,i} = \vec{K}_{4,i} \vec{L}_{4,k,i}^{\mu} \vec{B}_3^{y_{4,i}}
                \big\}_{i \in S \setminus \{k\}}.
            \end{align*}

        \item Let $S_?$ be the set of indexes that are delegatable fields
            in the vector $\vec{\sigma'}$. It selects random exponents $\{
            \tau_j, y_{1,j,u}, y_{1,j,h}, y_{2,j}, y_{3,j}, \{ y_{4,j,i}
            \}_{i \in S \cup \{j\}} \}_{j \in S_?} \in \Z_p$ and
            re-randomize the delegation components of the token as
            \begin{align*}
            \forall j \in S_? :~
            &   \vec{L}'_{1,j,u} = \vec{L}_{1,j,u}^{\mu} \vec{B}_3^{y_{1,j,u}},~
                \vec{L}'_{1,j,h} = \vec{L}_{1,j,h}^{\mu}
                    (\vec{L}_{1,k,u}^{\sigma_k} \vec{L}_{1,k,h})^{\tau_j}
                    \vec{B}_3^{y_{1,j,h}},~ \\
            &   \vec{L}'_{2,j} = \vec{L}_{2,j}^{\mu} \vec{L}_{2,j}^{\tau_j}
                    \vec{B}_3^{y_{2,j}},~
                \vec{L}'_{3,j} = \vec{L}_{3,j}^{\mu} \vec{L}_{3,j}^{\tau_j}
                    \vec{B}_3^{y_{3,j}},~ \\
            &   \vec{L}'_{4,j,j} = \vec{L}_{4,j,j}^{\mu} \vec{B}_3^{y_{4,j,j}},~
                \vec{L}'_{4,j,k} = \vec{L}_{4,j,k}^{\tau_j} \vec{B}_3^{y_{4,j,k}},~
                \big\{
                \vec{L}'_{4,j,i} = \vec{L}_{4,j,i}^{\mu} \vec{L}_{4,j,k}^{\tau_j}
                    \vec{B}_3^{y_{4,j,i}}
                \big\}_{i \in S \setminus \{k\}}.
            \end{align*}

        \item Finally, it outputs a token as
            \begin{align*}
            TK_{\vec{\sigma}'} = \Big(
                \vec{K}'_1,~ \vec{K}'_2,~ \vec{K}'_3,~
                \{ \vec{K}'_{4,i} \}_{i \in S},~
                \big\{
                \vec{L}'_{1,j,h}, \vec{L}'_{1,j,u},~
                \vec{L}'_{2,j},~ \vec{L}'_{3,j},~
                \{ \vec{L}'_{4,j,i} \}_{i \in S \cup \{j\}}
                \big\}_{j \in S_?}
            \Big).
            \end{align*}
        \end{enumerate}
    \end{enumerate}

\item [\tb{Encrypt}($\vec{x}, M, PK$):] It takes as input an attribute
    vector $\vec{x} = (x_1, \ldots, x_\ell) \in \Sigma^\ell$, a message $M
    \in \mathcal{M} \subseteq \G_T$, and the public key $PK$. It first
    chooses a random exponent $t \in \Z_p$ and random blinding values $z_1,
    z_2, z_3, \{ z_{4,i} \}_{i=1}^\ell \in \Z_p$. Then it outputs a
    ciphertext as
    \begin{align*}
    CT = \Big(~
        C_0 = \Omega^t M,~
        \vec{C}_1 = \vec{V}_c^t \vec{B}_2^{z_1},~
        \vec{C}_2 = \vec{W}_{c,1}^t \vec{B}_2^{z_2},~
        \vec{C}_3 = \vec{W}_{c,2}^t \vec{B}_2^{z_3},~
        \big\{
        \vec{C}_{4,i} = (\vec{U}_{c,i}^{x_i} \vec{H}_{c,i})^t \vec{B}_2^{z_{4,i}}
        \big\}_{i=1}^l
    ~\Big).
    \end{align*}

\item [\tb{Query}($CT, TK_{\vec{\sigma}}, PK$):] It takes as input a
    ciphertext $CT$ and a token $TK_{\vec{\sigma}}$ of a vector
    $\vec{\sigma}$. It first computes
    \begin{align*}
    M \leftarrow C_0 \cdot \Big( e(\vec{C}_1, \vec{K}_1) \cdot
        e(\vec{C}_2, \vec{K}_2) \cdot e(\vec{C}_3, \vec{K}_3) \cdot
        \prod_{i \in S} e(\vec{C}_{4,i}, \vec{K}_{4,i}) \Big)^{-1}.
    \end{align*}
    If $M \notin \mathcal{M}$, it outputs $\perp$ indicating that the
    predicate $f_{\vec{\sigma}}$ is not satisfied. Otherwise, it outputs
    $M$ indicating that the predicate $f_{\vec{\sigma}}$ is satisfied.
\end{description}

\subsection{Correctness}

If $f_{\vec{\sigma}}(\vec{x}) = 1$, then the following calculation shows that
$\textbf{Query}(CT, TK_{\vec{\sigma}}, PK) = M$ by the orthogonality of basis
vectors such that $e(g^{\vec{b}_{1,1}}, g^{\vec{b}_3}) = 1,
e(g^{\vec{b}_{1,2}}, g^{\vec{b}_2}) = 1, e(g^{\vec{b}_2}, g^{\vec{b}_3}) =
1$.
    \begin{align*}
    \lefteqn{ e(\vec{C}_1, \vec{K}_1) \cdot e(\vec{C}_2, \vec{K}_2) \cdot
        e(\vec{C}_3, \vec{K}_3) \cdot
        \prod_{i \in S} e(\vec{C}_{4,i}, \vec{K}_{4,i}) } \\
    &=  e(({\vec{V}_c})^t, \vec{B}_{1,2}^{\alpha} \vec{W}_{k,1}^{r_1} \vec{W}_{k,2}^{r_2}
            \prod_{i \in S} (\vec{U}_{c,i}^{\sigma_i} \vec{H}_{c,i})^{r_{3,i}}) \cdot
        e(\vec{W}_{c,1}^t, \vec{V}_k^{-r_1}) \cdot
        e(\vec{W}_{c,2}^t, \vec{V}_k^{-r_2}) \cdot
        \prod_{i \in S} e((\vec{U}_{c,i}^{x_i} \vec{H}_{c,i})^t, \vec{V}_k^{-r_{3,i}}) \\
    &=  e(\vec{B}_{1,1}^{v' t}, \vec{B}_{1,2}^{\alpha}) \cdot
        \prod_{i \in S} e(g^{v'}, g^{u'_i (\sigma_i - x_i)})^{t r_{3,i}}
     =  e(\vec{B}_{1,1}^{v'}, \vec{B}_{1,2})^{\alpha t}.
    \end{align*}
Otherwise, that is $f_{\vec{\sigma}}(\vec{x}) = 0$, the probability of
$\textbf{Query}(CT, TK_{\vec{\sigma}}, PK) \neq \perp$ is negligible by
limiting $|\mathcal{M}|$ to less than $|\G_T|^{1/4}$.

\subsection{Security}

\begin{theorem} \label{thm:sw-dhve}
The above dHVE scheme is selectively secure under the DBDH and P3DH
assumptions.
\end{theorem}

\begin{proof}
The proof of this theorem is easily obtained from the following five Lemmas
\ref{lem:sw-dhve}, \ref{lem:sw-dhve-a1}, \ref{lem:sw-dhve-a2},
\ref{lem:sw-dhve-a3} and \ref{lem:sw-dhve-a4}. Before presenting the five
lemmas, we first introduce the following four assumptions. The HVE scheme of
Shi and Waters constructed in bilinear groups of composite order $N = p_1 p_2
p_3$, and its security was proven under the DBDH, BSD, and C3DH assumptions
\cite{ShiW08}. In composite order bilinear groups, the C3DH assumption imply
the $l$-C3DH assumption that was introduced in \cite{ShiW08}. However, this
implication is not valid in prime order bilinear groups since the basis
vectors for ciphertexts and tokens are different. Thus the C3DH assumption
for ciphertexts and the C3DH assumption for tokens should be treated as
differently. These assumptions in composite order bilinear groups are
converted to the following Assumptions \ref*{sec:sw-dhve}-1,
\ref*{sec:sw-dhve}-2, \ref*{sec:sw-dhve}-3, and \ref*{sec:sw-dhve}-4 using
our conversion method.
\end{proof}

\vs \noindent \textbf{Assumption \ref*{sec:sw-dhve}-1} Let $((p, \G, \G_T,
e), g^{\vec{b}_{1,1}}, g^{\vec{b}_{1,2}}, g^{\vec{b}_2}, g^{\vec{b}_3})$ be
the bilinear product group of basis vectors $\vec{b}_{1,1} = (1,0,a_1),
\vec{b}_{1,2} = (1,a_2,0), \vec{b}_2 = (a_2,-1,a_1 a_2 - a_3), \vec{b}_3 =
(a_1, a_3, -1)$. The Assumption \ref*{sec:sw-dhve}-1 is stated as follows:
given a challenge tuple
    $$D = \big( (p, \G, \G_T, e),~
    g^{\vec{b}_{1,1}}, g^{\vec{b}_{1,2}}, g^{\vec{b}_2}, g^{\vec{b}_3},
    (g^{\vec{b}_{1,1}})^{c_1}, (g^{\vec{b}_{1,1}})^{c_2},
    (g^{\vec{b}_{1,2}})^{c_1}, (g^{\vec{b}_{1,2}})^{c_2},
    (g^{\vec{b}_{1,1}})^{c_3} \big) \mbox{ and } T,$$
decides whether $T = T_0 = e(g, g)^{c_1 c_2 c_3}$ or $T = T_1 = e(g,g)^d$
with random choices of $c_1, c_2, c_3, d \in \Z_p$.

\vs \noindent \textbf{Assumption \ref*{sec:sw-dhve}-2} Let $((p, \G, \G_T,
e), g^{\vec{b}_{1,1}}, g^{\vec{b}_{1,2}}, g^{\vec{b}_2}, g^{\vec{b}_3})$ be
the bilinear product group of basis vectors $\vec{b}_{1,1} = (1,0,a_1),
\vec{b}_{1,2} = (1,a_2,0), \vec{b}_2 = (a_2,-1,a_1 a_2 - a_3), \vec{b}_3 =
(a_1, a_3, -1)$. The Assumption \ref*{sec:sw-dhve}-2 is stated as follows:
given a challenge tuple
    $$D = \big( (p, \G, \G_T, e),~
    g^{\vec{b}_{1,1}}, g^{\vec{b}_{1,2}}, g^{\vec{b}_2}, g^{\vec{b}_3}
    \big) \mbox{ and } T,$$
decides whether $T = T_0 = e((g^{\vec{b}_{1,1}})^{c_1} (g^{\vec{b}_2})^{c_3},
(g^{\vec{b}_{1,2}})^{c_2} (g^{\vec{b}_3})^{c_4})$ or $T = T_1 =
e((g^{\vec{b}_{1,1}})^{c_1}, (g^{\vec{b}_{1,2}})^{c_2})$ with random choices
of $c_1, c_2, c_3, c_4 \in \Z_p$.

\vs \noindent \textbf{Assumption \ref*{sec:sw-dhve}-3} Let $((p, \G, \G_T,
e), g^{\vec{b}_{1,1}}, g^{\vec{b}_{1,2}}, g^{\vec{b}_2}, g^{\vec{b}_3})$ be
the bilinear product group of basis vectors $\vec{b}_{1,1} = (1,0,a_1),
\vec{b}_{1,2} = (1,a_2,0), \vec{b}_2 = (a_2,-1,a_1 a_2 - a_3), \vec{b}_3 =
(a_1, a_3, -1)$. The Assumption \ref*{sec:sw-dhve}-3 is stated as follows:
given a challenge tuple
    \begin{align*}
    D = \big( (p, \G, \G_T, e),~
    g^{\vec{b}_{1,1}}, g^{\vec{b}_{1,2}}, g^{\vec{b}_2}, g^{\vec{b}_3},
    (g^{\vec{b}_{1,2}})^{c_1}, (g^{\vec{b}_{1,2}})^{c_2},
    (g^{\vec{b}_{1,1}})^{c_1 c_2} (g^{\vec{b}_2})^{z_1},
    (g^{\vec{b}_{1,1}})^{c_1 c_2 c_3} (g^{\vec{b}_2})^{z_2}
    \big) \mbox{ and } T,
    \end{align*}
decides whether $T = T_0 = (g^{\vec{b}_{1,1}})^{c_3} (g^{\vec{b}_2})^{z_3}$
or $T = T_1 = (g^{\vec{b}_{1,1}})^d (g^{\vec{b}_2})^{z_3}$ with random
choices of $c_1, c_2, c_3, d \in \Z_p$, and $z_1, z_2, z_3 \in \Z_p$.

\vs \noindent \textbf{Assumption \ref*{sec:sw-dhve}-4} Let $((p, \G, \G_T,
e), g^{\vec{b}_{1,1}}, g^{\vec{b}_{1,2}}, g^{\vec{b}_2}, g^{\vec{b}_3})$ be
the bilinear product group of basis vectors $\vec{b}_{1,1} = (1,0,a_1),
\vec{b}_{1,2} = (1,a_2,0), \vec{b}_2 = (a_2,-1,a_1 a_2 - a_3), \vec{b}_3 =
(a_1, a_3, -1)$. The Assumption \ref*{sec:sw-dhve}-4 is stated as follows:
given a challenge tuple
    \begin{align*}
    D = \big( (p, \G, \G_T, e),~
    g^{\vec{b}_{1,1}}, g^{\vec{b}_{1,2}}, g^{\vec{b}_2}, g^{\vec{b}_3},
    (g^{\vec{b}_{1,1}})^{c_1}, (g^{\vec{b}_{1,1}})^{c_2},
    (g^{\vec{b}_{1,2}})^{c_1 c_2} (g^{\vec{b}_3})^{z_1},
    (g^{\vec{b}_{1,2}})^{c_1 c_2 c_3} (g^{\vec{b}_3})^{z_2}
    \big) \mbox{ and } T,
    \end{align*}
decides whether $T = T_0 = (g^{\vec{b}_{1,2}})^{c_3} (g^{\vec{b}_3})^{z_3}$
or $T = T_1 = (g^{\vec{b}_{1,2}})^d (g^{\vec{b}_3})^{z_3}$ with random
choices of $c_1, c_2, c_3, d \in \Z_p$, and $z_1, z_2, z_3 \in \Z_p$.

\begin{lemma} \label{lem:sw-dhve}
The above dHVE scheme is selectively secure under the Assumptions
\ref*{sec:sw-dhve}-1, \ref*{sec:sw-dhve}-2, \ref*{sec:sw-dhve}-3, and
\ref*{sec:sw-dhve}-4.
\end{lemma}

\begin{proof}
The proof of this lemma is directly obtained from \cite{ShiW08} since the
Assumptions \ref*{sec:sw-dhve}-1, \ref*{sec:sw-dhve}-2, \ref*{sec:sw-dhve}-3,
and \ref*{sec:sw-dhve}-4 in prime order bilinear groups are correspond to the
DBDH, BSD, C3DH (for ciphertexts), and C3DH (for tokens) assumptions in
composite order bilinear groups.
\end{proof}

\begin{lemma} \label{lem:sw-dhve-a1}
If the DBDH assumption holds, then the Assumption \ref*{sec:sw-dhve}-1 also
holds.
\end{lemma}

\begin{proof}
Suppose there exists an adversary $\mc{A}$ that breaks the Assumption
\ref*{sec:sw-dhve}-1 with a non-negligible advantage. An algorithm $\mc{B}$
that solves the DBDH assumption using $\mc{A}$ is given: a challenge tuple $D
= ((p, \G, \G_T, e), g, g^{c_1}, g^{c_2}, g^{c_3})$ and $T$ where $T = T_0 =
e(g,g)^{c_1 c_2 c_3}$ or $T = T_1 = e(g,g)^d$. $\mc{B}$ first chooses random
values $a_1, a_2, a_3 \in \Z_p$ and sets
    \begin{align*}
    &   g^{\vec{b}_{1,1}} = (g, 1, g^{a_1}),~
        g^{\vec{b}_{1,2}} = (g, g^{a_2}, 1),~
        g^{\vec{b}_2} = (g^{a_2}, g^{-1}, g^{a_1 a_2 - a_3}),~
        g^{\vec{b}_3} = (g^{a_1}, g^{a_3}, g^{-1}),~ \\
    &   (g^{\vec{b}_{1,1}})^{c_1} = (g^{c_1}, 1, (g^{c_1})^{a_1}),~
        (g^{\vec{b}_{1,1}})^{c_2} = (g^{c_2}, 1, (g^{c_2})^{a_1}),~
        (g^{\vec{b}_{1,1}})^{c_3} = (g^{c_3}, 1),~ \\
    &   (g^{\vec{b}_{1,2}})^{c_1} = (g^{c_1}, (g^{c_1})^{a_2}, 1),~
        (g^{\vec{b}_{1,2}})^{c_2} = (g^{c_2}, (g^{c_2})^{a_2}, 1).
    \end{align*}
Next, it gives the tuple
    $D' = ((p,\G, \G_T, e),
    g^{\vec{b}_{1,1}}, g^{\vec{b}_{1,2}}, g^{\vec{b}_2},
    (g^{\vec{b}_{1,1}})^{c_1}, (g^{\vec{b}_{1,1}})^{c_2},
    (g^{\vec{b}_{1,2}})^{c_1}, (g^{\vec{b}_{1,2}})^{c_2},
    (g^{\vec{b}_{1,1}})^{c_3})$ and $T$
to $\mc{A}$. Then $\mc{A}$ outputs a guess $\gamma'$. $\mc{B}$ also outputs
$\gamma'$. If the advantage of $\mc{A}$ is $\epsilon$, then the advantage of
$\mc{B}$ is greater than $\epsilon$ since the distribution of the challenge
tuple to $\mc{A}$ is equal to the Assumption \ref*{sec:sw-dhve}-1.
\end{proof}

\begin{lemma} \label{lem:sw-dhve-a2}
The Assumption \ref*{sec:sw-dhve}-2 holds for all adversaries.
\end{lemma}

\begin{proof}
The equation $e((g^{\vec{b}_{1,1}})^{c_1} (g^{\vec{b}_2})^{c_3},
(g^{\vec{b}_{1,2}})^{c_2} (g^{\vec{b}_3})^{c_4}) =
e((g^{\vec{b}_{1,1}})^{c_1}, (g^{\vec{b}_{1,2}})^{c_2})$ holds by the
orthogonality of basis vectors such that $e(g^{\vec{b}_{1,1}}, g^{\vec{b}_3})
= 1, e(g^{\vec{b}_2}, g^{\vec{b}_{1,2}}) = 1, e(g^{\vec{b}_2}, g^{\vec{b}_3})
= 1$. Therefore, any adversary can not break the Assumption
\ref*{sec:sw-dhve}-2.
\end{proof}

\begin{lemma} \label{lem:sw-dhve-a3}
If the P3DH assumption holds, then the Assumption \ref*{sec:sw-dhve}-3 also
holds.
\end{lemma}

\begin{proof}
Suppose there exists an adversary $\mc{A}$ that breaks the Assumption
\ref*{sec:sw-dhve}-3 with a non-negligible advantage. An algorithm $\mc{B}$
that solves the P3DH assumption using $\mc{A}$ is given: a challenge tuple
    $D = ((p, \G, \G_T, e), (g,f),
    (g^{c_1}, f^{c_1}), (g^{c_2}, f^{c_2}),
    (g^{c_1 c_2} f^{z_1}, g^{z_1}), (g^{c_1 c_2 c_3} f^{z_2}, g^{z_2}))$
and $T = T_{\gamma} = (T_{\gamma,1}, T_{\gamma,2})$ where $T = T_0 = (g^{c_3}
f^{z_3}, g^{z_3})$ or $T = T_1 = (g^d f^{z_3}, g^{z_3})$. $\mc{B}$ first
chooses random values $a_1, a_3 \in \Z_p$ and sets
    \begin{align*}
    &   g^{\vec{b}_{1,1}} = (g, 1, g^{a_1}),~
        g^{\vec{b}_{1,2}} = (g, f, 1),~
        g^{\vec{b}_2} = (f, g^{-1}, f^{a_1} g^{-a_3}),~
        g^{\vec{b}_3} = (g^{a_1}, g^{a_3}, g^{-1}),~ \\
    &   (g^{\vec{b}_{1,2}})^{c_1} = (g^{c_1}, f^{c_1}, 1),~
        (g^{\vec{b}_{1,2}})^{c_2} = (g^{c_2}, f^{c_2}, 1),~ \\
    &   (g^{\vec{b}_{1,1}})^{c_1 c_2} (g^{\vec{b}_2})^{z_1}
        = (g^{c_1 c_2} f^{z_1}, (g^{z_1})^{-1}, (g^{c_1 c_2} f^{z_1})^{a_1}
        (g^{z_1})^{-a_3}),~ \\
    &   (g^{\vec{b}_{1,1}})^{c_1 c_2 c_3} (g^{\vec{b}_2})^{z_2}
        = (g^{c_1 c_2 c_3} f^{z_2}, (g^{z_2})^{-1},
        (g^{c_1 c_2 c_3} f^{z_2})^{a_1} (g^{z_2})^{-a_3}),~ \\
    &   T' = (T_{\gamma,1}, T_{\gamma,2}, (T_{\gamma,1})^{a_1}
        (T_{\gamma,2})^{-a_3}).
    \end{align*}
Intuitively, it sets $a_2 = \log f$. Next, it gives the tuple
    $D' = ((p,\G, \G_T, e),
    g^{\vec{b}_{1,1}}, g^{\vec{b}_{1,2}}, g^{\vec{b}_2}, g^{\vec{b}_3},
    (g^{\vec{b}_{1,2}})^{c_1}, \lb (g^{\vec{b}_{1,2}})^{c_2}, \lb
    (g^{\vec{b}_{1,1}})^{c_1 c_2} (g^{\vec{b}_2})^{z_1},
    (g^{\vec{b}_{1,1}})^{c_1 c_2 c_3} (g^{\vec{b}_2})^{z_2} )$ and $T'$
to $\mc{A}$. Then $\mc{A}$ outputs a guess $\gamma'$. $\mc{B}$ also outputs
$\gamma'$. If the advantage of $\mc{A}$ is $\epsilon$, then the advantage of
$\mc{B}$ is greater than $\epsilon$ since the distribution of the challenge
tuple to $\mc{A}$ is equal to the Assumption \ref*{sec:sw-dhve}-3.
\end{proof}

\begin{lemma} \label{lem:sw-dhve-a4}
If the P3DH assumption holds, then the Assumption \ref*{sec:sw-dhve}-4 also
holds.
\end{lemma}

\begin{proof}
Suppose there exists an adversary $\mc{A}$ that breaks the Assumption
\ref*{sec:sw-dhve}-4 with a non-negligible advantage. An algorithm $\mc{B}$
that solves the P3DH assumption using $\mc{A}$ is given: a challenge tuple
    $D = ((p, \G, \G_T, e), (g,f),
    (g^{c_1}, f^{c_1}), (g^{c_2}, f^{c_2}),
    (g^{c_1 c_2} f^{z_1}, g^{z_1}), (g^{c_1 c_2 c_3} f^{z_2}, g^{z_2}))$
and $T = T_{\gamma} = (T_{\gamma,1}, T_{\gamma,2})$ where $T_0 = (g^{c_3}
f^{z_3}, g^{z_3})$ or $T_1 = (g^d f^{z_3}, g^{z_3})$. $\mc{B}$ first chooses
random values $a_2, a_3 \in \Z_p$ and sets
    \begin{align*}
    &   g^{\vec{b}_{1,1}} = (g, 1, f),~
        g^{\vec{b}_{1,2}} = (g, g^{a_2}, 1),~
        g^{\vec{b}_2} = (g^{a_2}, g^{-1}, g^{a_3}),~
        g^{\vec{b}_3} = (f, f^{a_2} g^{-a_3}, g^{-1}),~ \\
    &   (g^{\vec{b}_{1,1}})^{c_1} = (g^{c_1}, 1, f^{c_1}),~
        (g^{\vec{b}_{1,1}})^{c_2} = (g^{c_2}, 1, f^{c_2}),~ \\
    &   (g^{\vec{b}_{1,2}})^{c_1 c_2} (g^{\vec{b}_3})^{z_1}
        = (g^{c_1 c_2} f^{z_1}, (g^{c_1 c_2} f^{z_1})^{a_2} (g^{z_1})^{-a_3},
          (g^{z_1})^{-1}),~ \\
    &   (g^{\vec{b}_{1,2}})^{c_1 c_2 c_3} (g^{\vec{b}_3})^{z_2}
        = (g^{c_1 c_2 c_3} f^{z_2}, (g^{c_1 c_2 c_3} f^{z_2})^{a_2} (g^{z_2})^{-a_3},
          (g^{z_2})^{-1}),~ \\
    &   T' = (T_{\gamma,1}, (T_{\gamma,1})^{a_2} (T_{\gamma,2})^{-a_3},
        (T_{\gamma,2})^{-1}).
    \end{align*}
Intuitively, it sets $a'_1 = \log f, a'_2 = a_2, a'_3 = a_1 a_2 - a_3$ where
$a'_1, a'_2, a'_3$ are elements of basis vectors for the Assumption
\ref*{sec:sw-dhve}-4. Next, it gives the tuple
    $D' = ((p,\G, \G_T, e),
    g^{\vec{b}_{1,1}}, g^{\vec{b}_{1,2}}, g^{\vec{b}_2}, g^{\vec{b}_3},
    (g^{\vec{b}_{1,1}})^{c_1}, (g^{\vec{b}_{1,1}})^{c_2},
    (g^{\vec{b}_{1,2}})^{c_1 c_2} \cdot (g^{\vec{b}_3})^{z_1},
    (g^{\vec{b}_{1,2}})^{c_1 c_2 c_3} (g^{\vec{b}_3})^{z_2} )$ and $T'$
to $\mc{A}$. Then $\mc{A}$ outputs a guess $\gamma'$. $\mc{B}$ also outputs
$\gamma'$. If the advantage of $\mc{A}$ is $\epsilon$, then the advantage of
$\mc{B}$ is greater than $\epsilon$ since the distribution of the challenge
tuple to $\mc{A}$ is equal to the Assumption \ref*{sec:sw-dhve}-4.
\end{proof}

\section{Conversion 3: LL-HVE} \label{sec:ll-hve}

In this section, we convert the HVE scheme of Lee and Lee \cite{LeeL11} to
prime order bilinear groups and prove its selective security under the DBDH
and P3DH assumptions.

\subsection{Construction}

\begin{description}
\item [\tb{Setup}($1^{\lambda}, \ell$):] It generates the bilinear group
    $\G$ of prime order $p$ of bit size $\Theta(\lambda)$. It chooses
    random values $a_1, a_2, a_3 \in \Z_p$ and sets basis vectors for
    bilinear product groups as $\vec{b}_{1,1} = (1, 0, a_1),~ \vec{b}_{1,2}
    = (1, a_2, 0),~ \vec{b}_2 = (a_2, -1, a_1 a_2 - a_3),~ \vec{b}_3 =
    (a_1, a_3, -1)$. It also sets
    \begin{align*}
    &   \vec{B}_{1,1} = g^{\vec{b}_{1,1}},~
        \vec{B}_{1,2} = g^{\vec{b}_{1,2}},~
        \vec{B}_2 = g^{\vec{b}_2},~
        \vec{B}_3 = g^{\vec{b}_3}.
    \end{align*}
    It selects random exponents $v', w'_1, w'_2, \{ u'_i, h_i
    \}_{i=1}^\ell, \alpha \in \Z_p$, $z_v, z_{w,1}, z_{w,2}, \{ z_{u,i},
    z_{h,i} \}_{i=1}^\ell \in \Z_p$ and outputs a secret key and a public
    key as
    \begin{align*}
    SK = \Big(~
    &   V_k = \vec{B}_{1,2}^{v'},
        W_{k,1} = \vec{B}_{1,2}^{w'_1},
        W_{k,2} = \vec{B}_{1,2}^{w'_2},~
        \big\{
        U_{k,i} = \vec{B}_{1,2}^{u'_i},
        H_{k,i} = \vec{B}_{1,2}^{h'_i}
        \big\}_{i=1}^\ell,~
        \vec{B}_{1,2}^{\alpha}
        \Big), \\
    PK = \Big(~
    &   \vec{B}_{1,1},~ \vec{B}_{1,2},~ \vec{B}_2,~ \vec{B}_3,~
        \vec{V}_c = \vec{B}_{1,1}^{v'} \vec{B}_2^{z_v},~
        \vec{W}_{c,1} = \vec{B}_{1,1}^{w'_1} \vec{B}_2^{z_{w,1}},~
        \vec{W}_{c,2} = \vec{B}_{1,1}^{w'_2} \vec{B}_2^{z_{w,2}},~ \\
    &   \big\{
        \vec{U}_{c,i} = {\vec{B}_{1,1}^{u'_i}} \vec{B}_2^{z_{u,i}},~
        \vec{H}_{c,i} = {\vec{B}_{1,1}^{h'_i}} \vec{B}_2^{z_{h,i}}
        \big\}_{i=1}^\ell,~
        \Omega = e(\vec{B}_{1,1}^{v'}, \vec{B}_{1,2})^{\alpha} ~\Big).
    \end{align*}

\item [\tb{GenToken}($\vec{\sigma}, SK, PK$):] It takes as input a vector
    $\vec{\sigma} = (\sigma_1, \ldots, \sigma_\ell) \in \Sigma_*^\ell$ and
    the secret key $SK$. Let $S$ be the set of indexes that are not
    wild-card fields in the vector $\vec{\sigma}$. It selects random
    exponents $r_1, r_2, r_3 \in \Z_p$ and random blinding values $y_1,
    y_2, y_3, y_4 \in \Z_p$. Next it outputs a token as
    \begin{align*}
    TK_{\vec{\sigma}} = \Big(~
    &   \vec{K}_1 = \vec{B}_{1,2}^{\alpha} \vec{W}_{k,1}^{r_1} \vec{W}_{k,2}^{r_2}
            \prod_{i \in S} (\vec{U}_{k,i}^{\sigma_i} \vec{H}_{k,i})^{r_3}
            \vec{B}_3^{y_1},~
        \vec{K}_2 = \vec{V}_k^{-r_1} \vec{B}_3^{y_2},~
        \vec{K}_3 = \vec{V}_k^{-r_2} \vec{B}_3^{y_3},~
        \vec{K}_4 = \vec{V}_k^{-r_3} \vec{B}_3^{y_4}
    ~\Big).
    \end{align*}

\item [\tb{Encrypt}($\vec{x}, M, PK$):] It takes as input a vector $\vec{x}
    = (x_1, \ldots, x_\ell) \in \Sigma^l$, a message $M \in \mc{M}$, and
    the public key $PK$. It first chooses a random exponent $t \in \Z_p$
    and random blinding values $z_1, z_2, z_3, \{ z_{4,i} \}_{i=1}^\ell \in
    \Z_p$. Then it outputs a ciphertext as
    \begin{align*}
    CT = \Big(~
        C_0 = \Omega^t M,~
        \vec{C}_1 = \vec{V}_c^t \vec{B}_2^{z_1},~
        \vec{C}_2 = \vec{W}_{c,1}^t \vec{B}_2^{z_2},~
        \vec{C}_3 = \vec{W}_{c,2}^t \vec{B}_2^{z_3},~
        \big\{
        \vec{C}_{4,i} = (\vec{U}_{c,i}^{x_i} \vec{H}_{c,i})^t \vec{B}_2^{z_{4,i}}
        \big\}_{i=1}^\ell
    ~\Big).
    \end{align*}

\item [\tb{Query}($CT, TK_{\vec{\sigma}}, PK$):] It takes as input a
    ciphertext $CT$ and a token $TK_{\vec{\sigma}}$ of a vector
    $\vec{\sigma}$. It first computes
    \begin{align*}
    M \leftarrow C_0 \cdot \Big(
        e(\vec{C}_1, \vec{K}_1) \cdot e(\vec{C}_2, \vec{K}_2) \cdot
        e(\vec{C}_3, \vec{K}_3) \cdot
        e(\prod_{i \in S} \vec{C}_{4,i}, \vec{K}_4) \Big)^{-1}.
    \end{align*}
    If $M \notin \mathcal{M}$, it outputs $\perp$ indicating that the
    predicate $f_{\vec{\sigma}}$ is not satisfied. Otherwise, it outputs
    $M$ indicating that the predicate $f_{\vec{\sigma}}$ is satisfied.
\end{description}

\subsection{Correctness}

If $f_{\vec{\sigma}}(\vec{x}) = 1$, then the following calculation shows that
$\tb{Query}(CT, TK_{\vec{\sigma}}, PK) = M$ by the orthogonality of basis
vectors such that $e(g^{\vec{b}_{1,1}}, g^{\vec{b}_3}) = 1,
e(g^{\vec{b}_{1,2}}, g^{\vec{b}_2}) = 1, e(g^{\vec{b}_2}, g^{\vec{b}_3}) =
1$.
    \begin{align*}
    \lefteqn{ e(\vec{C}_1, \vec{K}_1) \cdot e(\vec{C}_2, \vec{K}_2) \cdot
        e(\vec{C}_3, \vec{K}_3) \cdot
        e(\prod_{i \in S} \vec{C}_{4,i}, \vec{K}_4) } \\
    &=  e(\vec{V}_c^t, \vec{B}_{1,2}^{\alpha} \vec{W}_{k,1}^{r_1} \vec{W}_{k,2}^{r_2}
            \prod_{i \in S} (\vec{U}_{k,i}^{\sigma_i} \vec{H}_{k,i})^{r_3}) \cdot
        e(\vec{W}_{c,1}^t, \vec{V}_k^{-r_1}) \cdot e(\vec{W}_{c,2}^t, \vec{V}_k^{-r_2}) \cdot
        e(\prod_{i \in S} (\vec{U}_{c,i}^{x_i} \vec{H}_{c,i})^t, \vec{V}_k^{-r_3}) \\
    &=  e(\vec{B}_{1,1}^{v' t}, \vec{B}_{1,2}^{\alpha}) \cdot
        e(g^{v'}, \prod_{i \in S} g^{u'_i (\sigma_i - x_i)})^{t r_3}
     =  e(\vec{B}_{1,1}^{v'}, \vec{B}_{1,2})^{\alpha t}.
    \end{align*}
Otherwise, that is $f_{\vec{\sigma}}(\vec{x}) = 0$, the probability of
$\tb{Query}(CT, TK_{\vec{\sigma}}, PK) \neq \perp$ is negligible by limiting
$|\mathcal{M}|$ to less than $|\G_T|^{1/4}$.

\subsection{Security}

\begin{theorem} \label{thm:ll-hve}
The above HVE scheme is selectively secure under the DBDH and P3DH
assumptions.
\end{theorem}

\begin{proof}
The proof of this theorem is easily obtained from the following five Lemmas
\ref{lem:ll-hve}, \ref{lem:ll-hve-a1}, \ref{lem:ll-hve-a2},
\ref{lem:ll-hve-a3} and \ref{lem:ll-hve-a4}. Before presenting the five
lemmas, we first introduce the following four assumptions. The HVE scheme of
Lee and Lee constructed in bilinear groups of composite order $N = p_1 p_2
p_3$, and its security was proven under the DBDH, BSD, and C3DH assumptions
\cite{ShiW08}. In composite order bilinear groups, the C3DH assumption imply
the C2DH assumption that was introduced in \cite{LeeL11}. However, this
implication is not valid in prime order bilinear groups since the basis
vectors for ciphertexts and tokens are different. Thus the C3DH assumption
for ciphertexts and the C2DH assumption for tokens should be treated as
differently. These assumptions in composite order bilinear groups are
converted to the following Assumptions \ref*{sec:ll-hve}-1,
\ref*{sec:ll-hve}-2, \ref*{sec:ll-hve}-3, and \ref*{sec:ll-hve}-4 using our
conversion method.
\end{proof}

\vs \noindent \textbf{Assumption \ref*{sec:ll-hve}-1} Let $((p, \G, \G_T, e),
g^{\vec{b}_{1,1}}, g^{\vec{b}_{1,2}}, g^{\vec{b}_2}, g^{\vec{b}_3})$ be the
bilinear product group of basis vectors $\vec{b}_{1,1} = (1,0,a_1),
\vec{b}_{1,2} = (1,a_2,0), \vec{b}_2 = (a_2,-1,a_1 a_2 - a_3), \vec{b}_3 =
(a_1, a_3, -1)$. The Assumption \ref*{sec:ll-hve}-1 is stated as follows:
given a challenge tuple
    $$D = \big( (p, \G, \G_T, e),~
    g^{\vec{b}_{1,1}}, g^{\vec{b}_{1,2}}, g^{\vec{b}_2}, g^{\vec{b}_3},
    (g^{\vec{b}_{1,1}})^{c_1}, (g^{\vec{b}_{1,1}})^{c_2},
    (g^{\vec{b}_{1,2}})^{c_1}, (g^{\vec{b}_{1,2}})^{c_2},
    (g^{\vec{b}_{1,1}})^{c_3} \big) \mbox{ and } T,$$
decides whether $T = T_0 = e(g, g)^{c_1 c_2 c_3}$ or $T = T_1 = e(g,g)^d$
with random choices of $c_1, c_2, c_3, d \in \Z_p$.

\vs \noindent \textbf{Assumption \ref*{sec:ll-hve}-2} Let $((p, \G, \G_T, e),
g^{\vec{b}_{1,1}}, g^{\vec{b}_{1,2}}, g^{\vec{b}_2}, g^{\vec{b}_3})$ be the
bilinear product group of basis vectors $\vec{b}_{1,1} = (1,0,a_1),
\vec{b}_{1,2} = (1,a_2,0), \vec{b}_2 = (a_2,-1,a_1 a_2 - a_3), \vec{b}_3 =
(a_1, a_3, -1)$. The Assumption \ref*{sec:ll-hve}-2 is stated as follows:
given a challenge tuple
    $$D = \big( (p, \G, \G_T, e),~
    g^{\vec{b}_{1,1}}, g^{\vec{b}_{1,2}}, g^{\vec{b}_2}, g^{\vec{b}_3}
    \big) \mbox{ and } T,$$
decides whether $T = T_0 = e((g^{\vec{b}_{1,1}})^{c_1} (g^{\vec{b}_2})^{c_3},
(g^{\vec{b}_{1,2}})^{c_2} (g^{\vec{b}_3})^{c_4})$ or $T = T_1 =
e((g^{\vec{b}_{1,1}})^{c_1}, (g^{\vec{b}_{1,2}})^{c_2})$ with random choices
of $c_1, c_2, c_3, c_4 \in \Z_p$.

\vs \noindent \textbf{Assumption \ref*{sec:ll-hve}-3} Let $((p, \G, \G_T, e),
g^{\vec{b}_{1,1}}, g^{\vec{b}_{1,2}}, g^{\vec{b}_2}, g^{\vec{b}_3})$ be the
bilinear product group of basis vectors $\vec{b}_{1,1} = (1,0,a_1),
\vec{b}_{1,2} = (1,a_2,0), \vec{b}_2 = (a_2,-1,a_1 a_2 - a_3), \vec{b}_3 =
(a_1, a_3, -1)$. The Assumption \ref*{sec:ll-hve}-3 is stated as follows:
given a challenge tuple
    \begin{eqnarray*}
    D = \big( (p, \G, \G_T, e),~
    g^{\vec{b}_{1,1}}, g^{\vec{b}_{1,2}}, g^{\vec{b}_2}, g^{\vec{b}_3},
    (g^{\vec{b}_{1,2}})^{c_1}, (g^{\vec{b}_{1,2}})^{c_2},
    (g^{\vec{b}_{1,1}})^{c_1 c_2} (g^{\vec{b}_2})^{z_1},
    (g^{\vec{b}_{1,1}})^{c_1 c_2 c_3} (g^{\vec{b}_2})^{z_2}
    \big) \mbox{ and } T,
    \end{eqnarray*}
decides whether $T = T_0 = (g^{\vec{b}_{1,1}})^{c_3} (g^{\vec{b}_2})^{z_3}$
or $T = T_1 = (g^{\vec{b}_{1,1}})^d (g^{\vec{b}_2})^{z_3}$ with random
choices of $c_1, c_2, c_3, d \in \Z_p$ and $z_1, z_2, z_3 \in \Z_p$.

\vs \noindent \textbf{Assumption \ref*{sec:ll-hve}-4} Let $((p, \G, \G_T, e),
g^{\vec{b}_{1,1}}, g^{\vec{b}_{1,2}}, g^{\vec{b}_2}, g^{\vec{b}_3})$ be the
bilinear product group of basis vectors $\vec{b}_{1,1} = (1,0,a_1),
\vec{b}_{1,2} = (1,a_2,0), \vec{b}_2 = (a_2,-1,a_1 a_2 - a_3), \vec{b}_3 =
(a_1, a_3, -1)$. The Assumption \ref*{sec:ll-hve}-4 is stated as follows:
given a challenge tuple
    \begin{eqnarray*}
    D = \big( (p, \G, \G_T, e),~
    g^{\vec{b}_{1,1}}, g^{\vec{b}_{1,2}}, g^{\vec{b}_2}, g^{\vec{b}_3},
    (g^{\vec{b}_{1,2}})^{c_1} (g^{\vec{b}_3})^{z_1},
    (g^{\vec{b}_{1,2}})^{c_2} (g^{\vec{b}_3})^{z_2}
    \big) \mbox{ and } T,
    \end{eqnarray*}
decides whether $T = T_0 = (g^{\vec{b}_{1,2}})^{c_1 c_2}
(g^{\vec{b}_3})^{z_3}$ or $T = T_1 = (g^{\vec{b}_{1,2}})^d
(g^{\vec{b}_3})^{z_3}$ with random choices of $c_1, c_2, d \in \Z_p$ and
$z_1, z_2, z_3 \in \Z_p$.

\begin{lemma} \label{lem:ll-hve}
The above HVE scheme is selectively secure under the Assumptions
\ref*{sec:ll-hve}-1, \ref*{sec:ll-hve}-2, \ref*{sec:ll-hve}-3, and
\ref*{sec:ll-hve}-4.
\end{lemma}

\begin{proof}
The proof of this lemma is directly obtained from \cite{LeeL11} since the
Assumptions \ref*{sec:ll-hve}-1, \ref*{sec:ll-hve}-2, \ref*{sec:ll-hve}-3,
and \ref*{sec:ll-hve}-4 in prime order bilinear groups are corresponds to the
DBDH, BSD, C3DH, and C2DH assumptions in composite order bilinear groups.
\end{proof}

\begin{lemma} \label{lem:ll-hve-a1}
If the DBDH assumption holds, then the Assumption \ref*{sec:ll-hve}-1 also
holds.
\end{lemma}

\begin{lemma} \label{lem:ll-hve-a2}
The Assumption \ref*{sec:ll-hve}-2 holds for all adversaries.
\end{lemma}

\begin{lemma} \label{lem:ll-hve-a3}
If the P3DH assumption holds, then the Assumption \ref*{sec:ll-hve}-3 also
holds.
\end{lemma}

\noindent The Assumptions \ref*{sec:ll-hve}-1, \ref*{sec:ll-hve}-2, and
\ref*{sec:ll-hve}-3 are the same as the Assumptions \ref*{sec:sw-dhve}-1,
\ref*{sec:sw-dhve}-2, and \ref*{sec:sw-dhve}-3. Thus we omits the proofs of
Lemmas \ref*{lem:ll-hve-a1}, \ref*{lem:ll-hve-a2}, \ref*{lem:ll-hve-a3}.

\begin{lemma} \label{lem:ll-hve-a4}
If the P3DH assumption holds, then the Assumption \ref*{sec:ll-hve}-4 also
holds.
\end{lemma}

\begin{proof}
Suppose there exists an adversary $\mc{A}$ that breaks the Assumption
\ref*{sec:ll-hve}-4 with a non-negligible advantage. An algorithm $\mc{B}$
that solves the P3DH assumption using $\mc{A}$ is given: a challenge tuple
    $D = ((p, \G, \G_T, e), (g,f),
    (g^{c_1}, f^{c_1}), (g^{c_2}, f^{c_2}),
    (g^{c_1 c_2} f^{z_1}, g^{z_1}), (g^{c_3} f^{z_2}, g^{z_2}))$
and $T = T_{\gamma} = (T_{\gamma,1}, T_{\gamma,2})$ where $T = T_0 = (g^{c_1
c_2 c_3} f^{z_3}, g^{z_3})$ or $T = T_1 = (g^d f^{z_3}, g^{z_3})$. $\mc{B}$
first chooses random values $a_2, a_3 \in \Z_p$ and sets
    \begin{align*}
    &   g^{\vec{b}_{1,1}} = (g, 1, f),~
        g^{\vec{b}_{1,2}} = (g, g^{a_2}, 1),~
        g^{\vec{b}_2} = (g^{a_2}, g^{-1}, g^{a_3}),~
        g^{\vec{b}_3} = (f, f^{a_2} g^{-a_3}, g^{-1}),~ \\
    &   (g^{\vec{b}_{1,2}})^{c'_1} (g^{\vec{b}_3})^{z_1}
        = (g^{c_1 c_2} f^{z_1}, (g^{c_1 c_2} f^{z_1})^{a_2} (g^{z_1})^{-a_3},
            (g^{z_1})^{-1}),~ \\
    &   (g^{\vec{b}_{1,2}})^{c'_2} (g^{\vec{b}_3})^{z_2}
        = (g^{c_3} f^{z_2}, (g^{c_3} f^{z_2})^{a_2} (g^{z_2})^{-a_3}, (g^{z_2})^{-1}),~ \\
    &   T' = (T_{\gamma,1}, (T_{\gamma,1})^{a_2} (T_{\gamma,2})^{-a_3},
        (T_{\gamma,2})^{-1}).
    \end{align*}
Intuitively, it sets $a'_1 = \log f, a'_2 = a_2, a'_3 = a_1 a_2 - a_3$ and
$c'_1 = c_1 c_2, c'_2 = c_3$ where $a'_1, a'_2, a'_3$ are elements of basis
vectors for the Assumption \ref*{sec:ll-hve}-4. Next, it gives the tuple
    $D' = ((p,\G, \G_T, e),
    g^{\vec{b}_{1,1}}, g^{\vec{b}_{1,2}}, g^{\vec{b}_2}, g^{\vec{b}_3}, \lb
    (g^{\vec{b}_{1,1}})^{c'_1} (g^{\vec{b}_2})^{z_1},
    (g^{\vec{b}_{1,1}})^{c'_2} (g^{\vec{b}_2})^{z_2} )$ and $T'$
to $\mc{A}$. Then $\mc{A}$ outputs a guess $\gamma'$. $\mc{B}$ also outputs
$\gamma'$. If the advantage of $\mc{A}$ is $\epsilon$, then the advantage of
$\mc{B}$ is greater than $\epsilon$ since the distribution of the challenge
tuple to $\mc{A}$ is equal to the Assumption \ref*{sec:ll-hve}-4.
\end{proof}

\section{Conclusion}

We converted the HVE scheme of Boneh and Waters, the delegatable HVE scheme of
Shi and Waters, and the efficient HVE scheme of Lee and Lee from composite
order bilinear groups to prime order bilinear groups. Though we used our
conversion method to HVE schemes that based on the decisional C3DH assumption,
it would be possible to use our method to other scheme in composite order
bilinear groups that based on the decisional C3DH assumption.

\bibliographystyle{plain}
\bibliography{convert-hve-prime}

\appendix

\section{Generic Group Model} \label{sec:gen-mod}

In this section, we show that the P3DH assumption holds in the generic group
model. The generic group model introduced by Shoup \cite{Shoup97} is a tool
for analyzing generic algorithms that work independently of the group
representation.

\subsection{Master Theorem}

We generalize the master theorem of Katz et al. \cite{KatzSW08} to use prime
order bilinear groups instead of composite order bilinear groups and to use
multiple groups elements in the target instead of just one element.

Let $\G, \G_T$ be cyclic bilinear groups of order $p$ where $p$ is a large
prime. The bilinear map is defined as $e:\G \times \G \rightarrow \G_T$. In
the generic group model, a random group element of $\G, \G_T$ is represented
as a random variable $P_i, R_i$ respectively where $P_i, R_i$ are chosen
uniformly in $\Z_p$. We say that a random variable has degree $t$ if the
maximum degree of any variable is $t$. Then we can naturally define the
dependence and independence of random variables as in Definition
\ref{def:dep-indep}.

\begin{definition} \label{def:dep-indep}
Let $P = \{P_1, \ldots, P_u\},~ T_0 = \{T_{0,1}, \ldots, T_{0,m}\},~ T_1 =
\{T_{1,1}, \ldots, T_{1,m}\}$ be random variables over $\G$ where $T_{0,i}
\neq T_{1,i}$ for all $1\leq i\leq m$, and let $R = \{R_1, \ldots, R_v\}$ be
random variables over $\G_T$. We say that $T_b$ is dependent on $A$ if there
exists constants $\{\alpha_i\}, \{\beta_i\}$ such that
    \begin{align*}
    \sum_i^m \alpha_i T_{b,i} = \sum_i^u \beta_i \cdot P_i
    \end{align*}
where $\alpha_i \neq 0$ for at least one $i$. We say that $T_b$ is independent
of $P$ if $T_b$ is not dependent on $P$.

Let $S_1 = \{ (i,j) ~|~ e(T_{0,i}, T_{0,j}) \neq e(T_{1,i}, T_{1,j}) \}$ and
$S_2 = \{ (i,j) ~|~ e(T_{0,i}, P_j) \neq e(T_{1,i}, P_j) \}$. We say that $\{
e(T_{b,i}, T_{b,j}) \}_{(i,j) \in S_1} \cup \{ e(T_{b,i},P_j) \}_{(i,j) \in
S_2}$ is dependent on $P \cup R \cup \{ e(T_{b,i}, T_{b,j}) \}_{(i,j) \notin
S_1} \cup \{ e(T_{b,i},P_j) \}_{(i,j) \notin S_2}$ if there exist constants
$\{\alpha_{i,j}\}, \{\alpha'_{i,j}\}, \{\beta_{i,j}\}, \{\beta'_{i,j}\},
\{\gamma_{i,j}\}, \{\delta_i\}$ such that
    \begin{align*}
    &   \sum_{(i,j) \in S_1} \alpha_{i,j} \cdot e(T_{b,i}, T_{b,j}) +
        \sum_{(i,j) \notin S_1} \alpha'_{i,j} \cdot e(T_{b,i}, T_{b,j}) +
        \sum_{(i,j) \in S_2} \beta_{i,j} \cdot e(T_{b,i}, P_j) +
        \sum_{(i,j) \notin S_2} \beta'_{i,j} \cdot e(T_{b,i}, P_j) \\
    &   = \sum_i^u \sum_j^u \gamma_{i,j} \cdot e(P_i, P_j) +
    \sum_i^v \delta_i \cdot R_i.
    \end{align*}
where $\alpha_{i,j} \neq 0$ for at least one $(i,j) \in S_1$ or $\beta_{i,j}
\neq 0$ for at least one $(i,j) \in S_2$. We say that $\{ e(T_{b,i}, T_{b,j})
\}_{(i,j) \in S_1} \cup \{ e(T_{b,i},P_j) \}_{(i,j) \in S_2}$ is independent
of $P \cup R \cup \{ e(T_{b,i}, T_{b,j}) \}_{(i,j) \notin S_1} \cup \{
e(T_{b,i},P_j) \}_{(i,j) \notin S_2}$ if $\{ e(T_{b,i}, T_{b,j}) \}_{(i,j)
\in S_1} \lb \cup \{ e(T_{b,i},P_j) \}_{(i,j) \in S_2}$ is not dependent on
$P \cup R \cup \{ e(T_{b,i}, T_{b,j}) \}_{(i,j) \notin S_1} \cup \{
e(T_{b,i},P_j) \}_{(i,j) \notin S_2}$.
\end{definition}

Using the above dependence and independence of random variables, we can
obtain the following theorem from the master theorem of Katz et al.
\cite{KatzSW08}.

\begin{theorem} \label{thm:master}
Let $P = \{P_1, \ldots, P_u\},~ T_0 = \{T_{0,1}, \ldots, T_{0,m}\},~ T_1 =
\{T_{1,1}, \ldots, T_{1,m}\}$ be random variables over $\G$ where $T_{0,i}
\neq T_{1,i}$ for all $1\leq i\leq m$, and let $R = \{R_1, \ldots, R_v\}$ be
random variables over $\G_T$. Consider the following experiment in the
generic group model:
    \begin{quote}
    An algorithm is given $P = \{P_1, \ldots, P_u\}$ and $R = \{R_1,
    \ldots, R_v\}$. A random bit $b$ is chosen, and the adversary is given
    $T_b = \{T_{b,1}, \ldots, T_{b,m}\}$. The algorithm outputs a bit $b'$,
    and succeeds if $b'=b$. The algorithm's advantage is the absolute value
    of the difference between its success probability and $1/2$.
    \end{quote}
Let $S_1 = \{(i,j) ~|~ e(T_{0,i}, T_{0,j}) \neq e(T_{1,i}, T_{1,j})\}$ and
$S_2 = \{ (i,j) ~|~ e(T_{0,i}, P_j) \neq e(T_{1,i}, P_j) \}$. If $T_b$ is
independent of $P$ for all $b \in \{0,1\}$, and $\{ e(T_{b,i}, T_{b,j})
\}_{(i,j) \in S_1} \cup \{ e(T_{b,i},P_j) \}_{(i,j) \in S_2}$ is independent
of $P \cup R \cup \{ e(T_{b,i}, T_{b,j}) \}_{(i,j) \notin S_1} \cup \{
e(T_{b,i},P_j) \}_{(i,j) \notin S_2}$ for all $b \in \{0,1\}$, then any
algorithm $\mc{A}$ issuing at most $q$ instructions has an advantage at most
$O(q^2t/p)$.
\end{theorem}

Note that this theorem that is a slight modification of that of Katz et al.
\cite{KatzSW08} still holds in prime order bilinear groups since the
dependent equation of an adversary can be used to distinguish the target
$T_b$ of the assumption. Additionally, it still holds when the target
consists of multiple group elements since the adversary can only make a
dependent equation in Definition \ref{def:dep-indep}.

\subsection{Analysis of P3DH Assumption}

To analyze the P3DH assumption in the generic group model, we only need to
show the independence of $T_0, T_1$ random variables. Using the notation of
previous section, the P3DH assumption can be written as follows
    \begin{align*}
    &   P = \{ 1, X, A, XA, B, XB, AB + XZ_1, Z_1, C + XZ_2, Z_2 \},~
        R = \{ 1 \} \\
    &   T_0 = \{ ABC + XZ_3, Z_3 \},~ T_1 = \{ D + XZ_3, Z_3 \}.
    \end{align*}

The $T_1$ has a random variable $D$ that does not exist in $P$. Thus the
independence of $T_1$ is easily obtained. Therefore, we only need to consider
the independence of $T_0$. First, $T_0$ is independent of $P$ since $T_0$
contains $Z_3$ that does not exist in $P$. For the independence of $\{
e(T_{0,i},T_{0,j}) \}_{(i,j) \in S_1} \cup \{ e(T_{0,i},P_j) \}_{(i,j) \in
S_2}$, we should define two sets $S_1, S_2$. We obtain that $S_1 = \{(1,1),
(1,2), (2,1), (2,2)\}$. However, $e(T_{0,i},T_{0,j})$ contains $Z_3^2$
because of $Z_3$ in $T_0$, and $Z_3^2$ can not be obtained from the right
part of the equation in Definition \ref{def:dep-indep}. Thus, the constants
$\alpha_{i,j}$ should be zero for all $(i,j)$. From this, we obtain the
simple equations as follows
    \begin{align*}
    \sum_{(i,j) \in S_2} \beta_{i,j} \cdot e(T_{b,i}, P_j) +
    \sum_{(i,j) \notin S_2} \beta'_{i,j} \cdot e(T_{b,i}, P_j)
    = \sum_i^u \sum_j^u \gamma_{i,j} \cdot e(P_i, P_j) +
    \sum_i^v \delta_i \cdot R_i.
    \end{align*}

The set $S_2$ is defined as $\{(i,j) ~|~ \forall i,j\}$ because of $D$ in
$T_1$. However, $Z_3$ in $T_0$ should be removed to construct a dependent
equation since $Z_3$ does not exists in $P,R$. To remove $Z_3$ from the left
part of the above simple equation, two random variables $Y, XY$ should be
paired with $T_{0,i}$ for some $Y \in P$. If $Z_3$ is remove in the left part
of the above simple equation, then the left part has at least a degree $3$ and
it contains $ABC$. To have a degree $3$ in the right part of the above simple
equation, $AB+XZ_1, Z_1$ should be used. However, the right part of the above
equation can not contain $ABC$ since $C, XC$ do not exist in $P$. Therefore,
the independence of $T_0$ is obtained.

\end{document}